\documentclass[prodmode,acmtalg]{acmsmall}
\usepackage[utf8]{inputenc}
\usepackage{microtype}
\usepackage{srcltx}

\let\accentvec\vec

\let\vec\accentvec

\usepackage{amsmath}
\usepackage{amssymb}
\usepackage{amsfonts}
\usepackage[textsize=tiny]{todonotes}

\usepackage{algorithm}
\usepackage[noend]{algpseudocode}

\newcommand{\E}{\textnormal{\textrm{E}}}

\newtheorem{claim}[theorem]{Claim}

\usepackage{pgfplots}
\newcommand{\Cpp}{C{}\texttt{++}{}}
\newcommand{\Samplesize}{\ensuremath{n_{\textnormal{\textrm{s}}}}}

\title{Optimal Partitioning for Dual-Pivot Quicksort}

\author{Martin Aum\"{u}ller and Martin
Dietzfelbinger \affil{Technische Universität Ilmenau%,
    }}
\begin{abstract}
    \emph{Dual-pivot quicksort} refers to variants of classical quicksort
		where in the partitioning step two pivots are used to split the input into three segments.
		This can be done in different ways, giving rise to different algorithms. 
		Recently, a dual-pivot algorithm proposed by Yaroslavskiy
    received much attention, because a variant of it replaced the
    well-engineered quicksort algorithm in Sun's Java 7 runtime library.
		Nebel and Wild (ESA 2012) analyzed this algorithm and showed that on average it 
		uses $1.9n\ln n +O(n)$ comparisons to sort an input of size $n$,
		beating standard quicksort, which uses $2n\ln n + O(n)$ comparisons.
		We introduce a model that captures all dual-pivot algorithms,
		give a unified analysis, and identify new dual-pivot algorithms that
		minimize the average number of key comparisons among all
		possible algorithms up to a linear term. 
		This minimum is $1.8n \ln n + O(n)$. For the case that the
		pivots are chosen from a small sample, we include a comparison
		of dual-pivot quicksort and classical quicksort. Specifically, 
                we show that dual-pivot quicksort benefits from a skewed choice of pivots.
                We experimentally evaluate our algorithms and compare them to
                Yaroslavskiy's algorithm and the recently described $3$-pivot quicksort
                algorithm of Kushagra et al. (ALENEX 2014).
\end{abstract}
\acmformat{Martin Aum\"{u}ller and Martin Dietzfelbinger. 2014. Optimal Partitioning for Dual-Pivot Quicksort.}
\begin{document}
\begin{bottomstuff}
    A preliminary version of this work appeared in the \emph{Proceedings of the 40th International Colloquium
    on Automata, Languages, and Programming (ICALP)}, Part 1, 33--44, 2013.

    Author's addresses: M. Aum\"{u}ller; M. Dietzfelbinger, Fakultät für Informatik und Automatisierung, Techni\-sche 
    Universität Ilmenau, 98683 Ilmenau, Germany; e-mail: \{martin.aumueller,martin.dietzfelbinger\}@tu-ilmenau.de
\end{bottomstuff}
\maketitle

\section{Introduction}\label{sec:introduction}

Quicksort~\cite{Hoare} is a thoroughly analyzed classical sorting algorithm,
described in standard textbooks such as~\cite{CLRS,Knuth,FlajSedg} 
and with implementations in practically all algorithm libraries. 
Following the divide-and-conquer paradigm, on an input consisting of $n$ elements
quicksort uses a pivot element to 
partition its input elements into two parts, 
the elements in one part being smaller than or equal to the pivot, the elements in the other 
part being larger than or equal to the pivot, and then uses recursion to sort these parts. 
It is well known that if the input consists of
$n$ elements with distinct keys in random order and the pivot is picked by just
choosing an element then on average quicksort uses $2n\ln n + O(n)$ comparisons.\footnote{
In this paper $\ln$ denotes the natural logarithm and $\log$ denotes the logarithm to base $2$.}
In 2009, Yaroslavskiy announced%
\footnote{An archived version of the relevant discussion in a Java newsgroup can be found at\\
    \texttt{http://permalink.gmane.org/gmane.comp.java.openjdk.core-libs.devel/2628}. Also see~\cite{nebel12}.} 
that he had found an improved quicksort implementation, the claim being backed by experiments. 
After extensive empirical studies, in 2009 a variant of Yaroslavskiy's
algorithm due to Yaroslavskiy, Bentley, and Bloch became the new standard quicksort algorithm in Sun's Java 7 runtime library. 
This algorithm employs two pivots to split the elements.
If two pivots $p$ and $q$ with $p \leq q$ are used, 
the partitioning step partitions the remaining $n-2$ elements into 3 parts:
elements smaller than or equal to $p$, 
elements between $p$ and $q$, and elements larger than or equal to $q$,
see~Fig.~\ref{fig:dual:pivot:partition}.%
\footnote{In accordance with tradition, we assume in this theoretical study that all elements have different keys. 
Of course, in implementations equal keys are an important issue that requires a lot of care~\cite{SedgewickEqual}.}
Recursion is then applied to the three parts.
%The outline of a dual-pivot quicksort algorithm is depicted in
%Algorithm~\ref{algo:dual:pivot:outline}.
As remarked in \cite{nebel12}, it came as a surprise that two pivots should help, since in his thesis~\cite{sedgewick}
Sedgewick had proposed and analyzed  a dual-pivot
approach that was inferior to classical quicksort. Later, Hennequin in his
thesis~\cite{hennequin}
studied the general approach of using $k \geq 1$ pivot elements. According to \cite{nebel12}, he found only
slight improvements that would not compensate for the more involved partitioning procedure. 
(See \cite{nebel12} for a short discussion.)

\begin{figure}[tb]
    \centering
    \scalebox{0.5}{\includegraphics{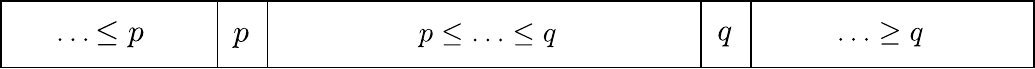}}
    \caption{Result of the partition step in dual-pivot quicksort schemes using
    two pivots $p,q$ with $p \leq q$. Elements left of $p$ are smaller than or equal to $p$, 
    elements right of $q$ are larger than or equal to $q$. The elements between $p$ and $q$
    are at least as large as $p$ and at most as large as $q$.}
    \label{fig:dual:pivot:partition}
\end{figure}

In \cite{nebel12} (see also the full version \cite{WildNN15}), Nebel and Wild
formulated and analyzed a simplified version
of Yaroslavskiy's algorithm. (For completeness, this algorithm is given as
Algorithm~\ref{algo:yaroslavskiy:partition} in
Appendix~\ref{app:sec:yaroslavskiy}.) They showed that it makes $1.9n \ln n +
O(n)$ key comparisons on average, in contrast to the $2n\ln n+O(n)$ of standard
quicksort and the $\frac{32}{15}n \ln n + O(n)$ of Sedgewick's dual-pivot
algorithm. On the other hand, they showed that the number of swap operations in
Yaroslavskiy's algorithm  is $0.6n \ln n + O(n)$ on average, which is much
higher than the $0.33 n\ln n + O(n)$ swap operations in classical quicksort.  In
this paper we concentrate on the comparison count as
cost measure and on asymptotic results. We leave the study of other cost measures
to further investigations (which already have taken place, see \cite{MartinezNW15}).
We consider other measures in experimental evaluations.

The authors of \cite{nebel12} state that the reason for Yaroslavskiy's algorithm
being superior were that his ``partitioning method is able to take advantage of
certain asymmetries in the
outcomes of key comparisons''. They also state that ``[Sedgewick's dual-pivot method] fails to utilize them, even though
being based on the same abstract algorithmic idea''.  
So the abstract algorithmic idea of using two pivots can lead to different algorithms with different behavior. 
In this paper we describe the design space from which all these algorithms originate.
We fully explain which simple property makes some dual-pivot algorithms perform
better and some perform worse w.r.t. the average comparison count
and identify optimal members (up to a linear term) of this design
space. The best ones use $1.8n \ln n
+ O(n)$ comparisons on average---even less than Yaroslavskiy's
method.

The first observation is that everything depends on the cost, i.\,e., the comparison count, of the partitioning
step.
This is not new at all. Actually, in Hennequin's thesis~\cite{hennequin}
the connection between partitioning cost and overall cost 
for quicksort variants with more than one pivot is analyzed in detail.
For dual-pivot quicksort, Hennequin proved that 
if the (average) partitioning cost for $n$ elements is $a\cdot n+O(1)$, for a constant $a$, then the average cost for sorting $n$ elements is 
\begin{equation}\label{eq:1}
    \frac65a\cdot n\ln n + O(n).
\end{equation}
The partitioning cost of some algorithms presented in this paper will have a non-constant lower order term.
Utilizing the Continuous Master Theorem of \cite{Roura01},
we prove that \eqref{eq:1} describes the average cost even if the partitioning cost is $a \cdot n + O(n^{1-\varepsilon})$ for 
some $\varepsilon > 0$.
Throughout the present paper all that interests us is the constant factor with the leading term.
(The reader should be warned that for real-life $n$ the linear term,
which can even be negative, can have a big influence
on the average number of comparisons.)

The second observation is that the partitioning cost depends on certain details of the 
partitioning procedure. This is in contrast to standard quicksort with one pivot
where partitioning always takes $n-1$ comparisons. In \cite{nebel12} it is shown
that Yaroslavskiy's partitioning procedure uses $\frac{19}{12}n + O(1)$ comparisons on average,
while Sedgewick's uses $\frac{16}{9}n + O(1)$ many. The analysis of these two
algorithms is based on the study of how certain pointers move through the array,
at which positions elements are compared to the pivots, which of the two pivots
is used for the first comparison, and how swap operations
exchange two elements in the array.
For understanding what is going on, however, it is helpful to forget about 
concrete implementations with loops in which pointers sweep across arrays and entries are swapped,
and look at partitioning with two pivots in a more abstract way.
For simplicity we shall always assume that the input is a permutation of $\{1,\ldots,n\}$.
Now pivots $p$ and $q$ with $p<q$ are chosen. 
The task is to \emph{classify} the remaining $n-2$ elements into classes 
``small'' ($s=p-1$ many), ``medium'' ($m=q-p-1$ many), and ``large'' ($\ell=n-p$ many), 
by comparing these elements one after the other with the smaller pivot or
the larger pivot, or both of them if necessary. 
Note that for symmetry reasons it is inessential in which order the elements are treated.
The only choice the algorithm can make is
whether to compare the current element with the smaller pivot or the larger pivot first. 
%(In Sedgewick's and in Yaroslavskiy's
%algorithm this decision is based on the position of certain pointers
%and the state of control.) 
Let the random variable $S_2$ denote the number of small elements compared with the larger pivot first,
and let $L_2$ denote the number of large elements compared with the smaller pivot first.
Then the total number of comparisons is 
$n-2 + m + S_2 + L_2.$

Averaging over all inputs and all possible choices of the pivots 
the term $n-2 + m$ will lead to $\frac43n + O(1)$ key comparisons on average,
independently of the algorithm. 
Let $W = S_2 +L_2$,
the number of elements that are compared with the ``wrong'' pivot first.
Then $\E(W)$ is the only quantity that is influenced by a particular partitioning procedure.

In the paper, we will first devise an easy method to calculate $\E(W)$. 
The result of this analysis will lead to an asymptotically optimal strategy.
The basic approach is the following. Assume a partitioning procedure is given, and 
assume $p,q$ and hence $s=p-1$ and $\ell=n-q$ are fixed, 
and let $w_{s,\ell} = \E(W \mid s,\ell)$.
Denote the average number of elements compared to the smaller [larger] pivot first by $f^\text{p}_{s,\ell}$ [$f^\text{q}_{s,\ell}$].
If the elements to be classified were chosen to be small, medium, and large independently
with probabilities $s/(n-2)$, $m/(n-2)$, and $\ell/(n-2)$, resp., then the
average number of 
small elements compared with the large pivot first would be $f^\text{q}_{s,\ell} \cdot s/(n-2)$,
similarly for the large elements. 
Of course, the actual input is a sequence with exactly $s$ [$m$, $\ell$] small [medium, large] elements, and there is no independence. 
Still, we will show that the randomness in the order is sufficient to guarantee that, for some $\varepsilon > 0$,
\begin{equation}
w_{s,\ell} = f^\text{q}_{s,\ell} \cdot s/n + f^\text{p}_{s,\ell} \cdot \ell/n + O(n^{1 - \varepsilon}).
\label{eq:small:large:200}
\end{equation}
The details of the partitioning procedure will determine $f^\text{p}_{s,\ell}$ and $f^\text{q}_{s,\ell}$, and hence
$w_{s,\ell}$ up to $O(n^{1 - \varepsilon})$.
This seemingly simple insight has two consequences, one for the analysis and one for the design of dual-pivot algorithms: 
\begin{itemize}
	\item[(i)] In order to \emph{analyze} the average comparison count of a
	    dual-pivot algorithm (given by its partitioning procedure) up to a linear term,
	determine $f^\text{p}_{s,\ell}$ and $f^\text{q}_{s,\ell}$ for this partitioning procedure. This will give
    $w_{s,\ell}$ up to $O(n^{1-\varepsilon})$, which must then be averaged over all $s,\ell$ to find the
	average number of comparisons in partitioning. Then apply \eqref{eq:1}.
	\item[(ii)] In order to \emph{design} a good partitioning procedure w.r.t.
	    the average comparison count, try to
	    make $f^\text{q}_{s,\ell} \cdot s/n + f^\text{p}_{s,\ell} \cdot \ell/n$ small.
\end{itemize}
We shall demonstrate approach (i) in Section~\ref{sec:methods}. An example: 
As explained in \cite{nebel12}, if $s$ and $\ell$ are fixed,  
in Yaroslavskiy's algorithm we have $f^\text{q}_{s,\ell} \approx \ell$ and $f^\text{p}_{s,\ell} \approx s+m$. 
By \eqref{eq:small:large:200} we get $w_{s,\ell} = (\ell s + (s+m) \ell)/n + O(n^{1 - \varepsilon})$.
This must be averaged over all possible values of $s$ and $\ell$. The result is
$\frac{1}{4}n + O(n^{1 - \varepsilon})$,
which together with $\frac{4}{3}n + O(1)$ gives $\frac{19}{12}n + O(n^{1 - \varepsilon})$, close
to the result  from \cite{nebel12}. 

Principle (ii) will be used to identify an asymptotically optimal partitioning procedure 
that makes $1.8n \ln n + O(n)$
key comparisons on average. 
In brief, such a strategy should achieve the following: 
	If $s>\ell$, compare (almost) all entries with the smaller pivot first ($f^\text{p}_{s,\ell}
	\approx n$ and $f^\text{q}_{s,\ell}\approx 0$),
	otherwise compare (almost) all entries with the larger pivot first
	($f^\text{p}_{s,\ell}\approx 0$ and $f^\text{q}_{s,\ell}\approx n$). 
	Of course, some details have to be worked out:
	How can the algorithm decide which case applies? In which technical sense is this strategy optimal?
	We shall see in Section~\ref{sec:decreasing} how a sampling technique resolves these issues.

        In 
Section~\ref{sec:optimal:strategies}, we will consider the following simple and intuitive strategy:
\emph{If more elements have been classified as being large instead of being small so far,
compare the next element to the larger pivot first, otherwise compare it to the smaller
pivot first.}
We will show
that this strategy is optimal w.r.t. minimizing the average comparison count.

In implementations of quicksort, the pivot is usually chosen as the median 
from a small sample of $2k + 1$ elements. Intuitively, this yields more balanced subproblems, 
which are smaller on average.
This idea already appeared in Hoare's original publication \cite{Hoare} without 
an analysis. This was later supplied by van Emden \cite{vanEmden}. The complete analysis of
this variant was given by Martinez and Roura in \cite{MartinezR01} in 2001. 
They showed that the optimal sample size is $\sqrt{n}$. For this
sample size the 
average comparison count of quicksort
matches the lower-order bound of $n \log n + O(n)$ comparisons. 
In practice, one usually uses a sample of size $3$. Theoretically, this decreases
the average comparison count from $2n \ln n + O(n)$ to $1.714..n \ln n + O(n)$.
This strategy has been generalized in the obvious way to Yaroslavskiy's
algorithm as well. The implementation of Yaroslavskiy's algorithm in Sun's Java 7 uses
the two tertiles in a sample of size $5$ as pivots, i.e., the elements of rank $2$ and $4$.
In
Section~\ref{sec:pivot:sample} we will analyze the comparison count of dual-pivot
quicksort algorithms with this sampling strategy. 
Yaroslavskiy's algorithm has an average comparison count of $1.704..n \ln n + O(n)$ 
in this case, while the optimal average cost is $1.623..n \ln n + O(n)$.
In that section, we will also consider a question raised by Wild, Nebel, and Mart\'inez
in \cite[Section 8]{WildNM14} for Yaroslavskiy's algorithm, namely of which rank the
pivots should be to achieve  minimum sorting cost. While using
the tertiles of the input 
seems the obvious choice for balancing reasons, in \cite{WildNM14} it is shown
that for Yaroslavskiy's algorithm this minimum is attained for ranks
$\approx 0.429n$ and $\approx
0.698n$ and is $\approx 1.4931 n \ln n$. We will show that the simple strategy \emph{``Always
  compare with the larger pivot first''} achieves sorting cost $\approx 1.4427 n \ln n$,
  i.\,e., the lower bound for comparison-based sorting,  when choosing the
  elements of rank $n/4$ and $n/2$ as the two pivots.

As noted in \cite{nebel13}, considering only key comparisons
and swap operations does not suffice for evaluating the practicability of
sorting algorithms. In Section~\ref{sec:experiments}, we will present
experimental results that indicate the following: When sorting integers, the
comparison-optimal algorithms of Section~\ref{sec:decreasing} and
Section~\ref{sec:optimal:strategies} are slower than Yaroslavskiy's
algorithm. However, an implementation of the simple strategy \emph{``Always compare
with the larger pivot first''} performs very well both in {\Cpp} and in Java in
our experimental setup.  We will also compare our algorithms to the fast
three-pivot quicksort algorithm described in \cite{Kushagra14}.
While comparing these algorithms, we will provide evidence that the theoretical cost measure ``cache misses'' 
described in \cite{Kushagra14} nicely predicts empirical cache behavior,
and comes closest for correctly predicting running time.

We emphasize that the purpose of this paper is not to arrive at better and
better quicksort algorithms by using all kinds of variations, but rather to
thoroughly analyze the situation with two pivots, showing the potential and the
limitations of this approach.  

\section{Basic Approach to Analyzing Dual-Pivot
Quicksort}\label{sec:average:case:analysis}
We assume the input sequence $(a_1,\ldots,a_n)$ to be a random permutation of $\{1,\ldots,n\}$,
each permutation occurring with probability $(1/n!)$. If $n \leq 1$, there is
nothing to do; if $n = 2$, sort by one comparison. Otherwise, choose the first element $a_1$ and
the last element $a_n$ as the set of pivots, and set $p = \min(a_1,a_n)$
and $q = \max(a_1,a_n)$. Partition the remaining elements into elements smaller
than $p$
(``small'' elements), 
elements between $p$ and $q$ (``medium'' elements), and elements larger than $q$
(``large'' elements), see Fig.~\ref{fig:dual:pivot:partition}.
Then apply the procedure recursively to these three 
groups.
Clearly, each pair $p,q$ with $1 \leq p < q \leq n$
appears as set of pivots with probability $1/\binom{n}{2}$.
Our cost measure is the number of key comparisons needed to
sort the given input. Let
$C_n$ be the random variable counting this number. Let $P_n$ denote the
partitioning cost to partition the $n-2$ non-pivot elements into the three
groups. As
explained by Wild and Nebel \cite[Appendix A]{nebel12}, the average number of
key comparisons obeys the following recurrence: 
\begin{align}\label{eq:25}
    \E(C_n) &= 
%    \sum_{1 \leq p < q \leq n} \Pr(\text{$p, q$ pivots})
%    \cdot \E(P_n + \text{ recursive costs} \mid p,q )\\
%    &= \sum_{1 \leq p < q \leq n} \frac{2}{n(n-1)}\cdot \E(P_n
%    + C_{p-1} + C_{q-p-1} + C_{n-q} \mid p,q)\\
%    &= 
    \E(P_n) + \frac{2}{n(n-1)} \cdot 3 \sum_{k = 0}^{n-2}
    (n-k-1)\cdot \E(C_k).
\end{align}
If $\E(P_n) = a \cdot n + O(1)$, for a constant $a$, this can
be solved (\emph{cf.} \cite{hennequin,nebel12}) to give
\begin{equation}\label{eq:10}
    \E(C_n) = \frac{6}{5}a \cdot n \ln n + O(n).
\end{equation}
In Section~\ref{sec:additional:cost:term} we will show that \eqref{eq:10} also holds
if $\E(P_n) = a \cdot n + O(n^{1-\varepsilon})$ for some $\varepsilon > 0$. 
For the proof we utilize the Continuous Master Theorem from \cite{Roura01}.

In view of this simple relation it is sufficient to study the cost of partitioning.
Abstracting from moving elements around in arrays, we arrive at the following
``classification problem'':
Given a random permutation $(a_1,\ldots,a_n)$ of $\{1,\ldots,n\}$ as the input
sequence and $a_1$ and $a_n$ as the two pivots $p$ and $q$, with $p < q$,
classify each of the remaining $n-2$ elements as being small, medium, or large.
Note that there are exactly
$s:=p-1$ small elements, $m := q - p - 1$ medium elements, and $\ell:=n-q$ large
elements. 
Although this
classification does not yield an actual partition of the input sequence, 
a classification algorithm can be turned into a partitioning algorithm using
only swap operations
but no additional key comparisons. Since elements are only
compared with the two pivots, the randomness of subarrays is preserved. Thus,
in the recursion we may always assume that the input is arranged randomly.

We make the following observations (and fix notation) for all classification
algorithms.  One key comparison is needed to decide which of the elements $a_1$
and $a_n$ is the smaller pivot $p$ and which is the larger pivot $q$.  For
classification, each of the remaining $n-2$ elements has to be compared against
$p$ or $q$ or both. Each \emph{medium} element has to be compared to $p$
\emph{and} $q$. On average, there are $(n-2)/3$ medium elements. Let $S_2$
denote the number of small elements that are compared to the larger pivot first,
i.\,e., the number of small elements that need $2$ comparisons for classification.
Analogously, let $L_2$ denote the number of large elements compared to the
smaller pivot first. Conditioning on the pivot choices, and hence the values of
$s$ and $\ell$, we may calculate $\E(P_n)$ as follows:\footnote{We use
the following notation throughout this paper: To indicate that sums run over all 
$\binom{n}{2}$ combinations $(s,\ell)$ with $s, \ell \geq 0$ and $s + \ell \leq n - 2$
we simply write $\sum_{s + \ell \leq n - 2}$.}
\begin{align}\label{eq:30}
    \E(P_n) &= (n - 1) + (n-2)/3 + \frac{1}{\binom{n}{2}} \sum_{s+\ell \leq n-2}\E(S_2 + L_2 \mid s, \ell).
\end{align}
We call the third summand in \eqref{eq:30} the \emph{additional cost term (ACT)}, as it is the only
value that depends on the actual classification algorithm.

\section{Analyzing the Additional Cost Term}\label{sec:additional:cost:term}

We will use the following formalization of a partitioning
procedure: A \emph{classification
strategy} is given as a three-way
decision tree $T$ with a root and $n-2$ levels numbered $0, 1, \ldots, n - 3$
of inner nodes as well as one
leaf level. The root is on level $0$. Each node $v$ is labeled with an index
$i(v)
\in \{2,\ldots, n - 1\}$ and an element $l(v) \in \{\text{p}, \text{q}\}$. If
$l(v)$ is $\text{p}$, then at node $v$ element $a_{i(v)}$ is compared with the
smaller pivot first; otherwise, i.\,e., $l(v) = \text{q}$, it is compared with the
larger pivot first. 
 The three edges out of a node are labeled $\sigma, \mu, \lambda,$ resp.,
representing the outcome of the classification as small, medium, large,
respectively. The label of edge $e$ is called $c(e)$. 
On each of the $3^{n-2}$ paths each index occurs exactly
once.  Each input
determines exactly one path $w$ from the root to a leaf in the
obvious way; the
classification of the elements can then be read off from the node and edge
labels along this path. We call such a tree a \emph{classification tree}.
 \newcommand{\lv}[1]{\ensuremath{\textnormal{\textrm{level}}(#1)}}

\begin{figure}[bt]
    \centering
    \scalebox{0.65}{\includegraphics{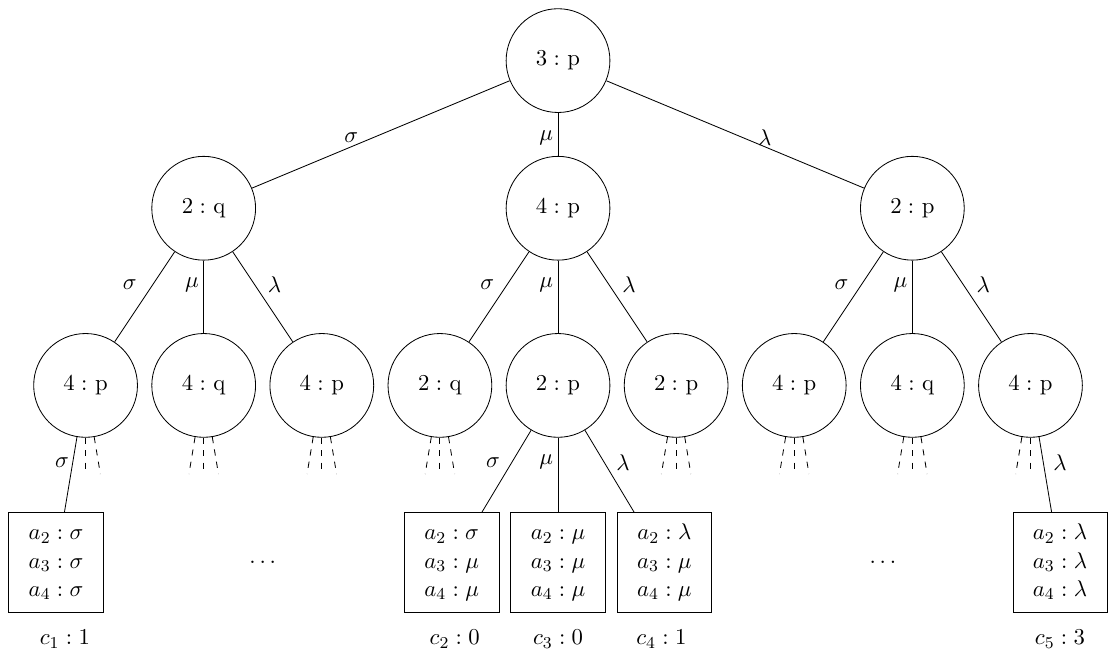}}
    \caption{An example for a decision tree to classify three elements $a_2,a_3,$ 
    and $a_4$ according to the pivots $a_1$ and $a_5$. 
    Five out of the $27$ leaves are explicitly drawn, showing the classification of the elements and 
    the costs $c_i$ of the specific paths.}
    \label{fig:decision:tree:ex1}
\end{figure}

Identifying a path $\pi$ from the root to a leaf $\text{lf}$ by the sequence of nodes and 
edges $(v_1,e_1,v_2,e_2,\ldots,v_{n-2}, e_{n-2}, \text{lf})$ on it, we define the cost $c_\pi$ as 
\begin{align*}
    c_\pi = \big|\big\{j \in
    \{1,\ldots,n-2\} \mid l(v_j) = 
\text{q} \wedge c(e_j) = \sigma  \text{ or } l(v_j) = \text{p} \wedge c(e_j) = \lambda \big\}\big|.
\end{align*}
For a given input, the cost of the path associated with this input exactly describes the number of additional
comparisons on this input. An example for such a classification tree is given in Figure~\ref{fig:decision:tree:ex1}.

We let $S^T_2$ [$L^T_2$] denote the random variable that for a random input
counts the number of small [large] elements classified in nodes with label
$\text{q}$ [$\text{p}$].  We now describe how we can calculate the ACT of a
classification tree $T$. First consider fixed $s$ and $\ell$ and let the input
excepting the pivots be arranged randomly.  For a node $v$ in $T$,  we let
$s_v$, $m_v$, and $\ell_v$, resp., denote the number of edges labeled $\sigma$,
$\mu$, and $\lambda$, resp., from the root to $v$. By the randomness of the
input, the probability that the element classified at $v$ is ``small'', i.\,e.,
that the edge labeled $\sigma$ is used,  is exactly $(s - s_v)/(n   - 2 -
\lv{v})$. The probability that it is ``medium'' is $(m - m_v)/(n   - 2 -
\lv{v})$, and that it is ``large'' is $(\ell - \ell_v)/(n - 2 - \lv{v})$. The
probability $p^v_{s,\ell}$ that node $v$ in the tree is reached is then just
the product of all these edge probabilities on the unique path from the root to
$v$.  The probability that the edge labeled $\sigma$ out of a node $v$ is used
can then be calculated as $p^v_{s,\ell} \cdot (s - s_v)/(n   - 2 - \lv{v})$.
Similarly, the probability that the edge labeled $\lambda$ is used is
$p^v_{s,\ell} \cdot (\ell - \ell_v)/(n   - 2 - \lv{v})$. Note that all this is
independent of the actual ordering in which the classification tree inspects
the elements. We can thus always assume some fixed ordering and forget about
the label $i(v)$ of node $v$.
    
 By linearity of expectation, we can sum up the
    contribution to the additional comparison count for each node separately.
    Thus, we may calculate
   \begin{align}\label{eq:35}
    \E(S^T_2 {+} L^T_2 \mid s, \ell) &=\sum_{\substack{v \in T\\l(v) =
   \text{q}}} p^v_{s,\ell} \cdot 
\frac{s - s_v}{n
- 2 - \lv{v}} + \sum_{\substack{v \in T\\l(v) = \text{p}}} p^v_{s,\ell} \cdot \frac{\ell -
    \ell_v}{n - 2 -
    \lv{v}}.
    \end{align}
The setup developed so far makes it possible to describe the connection between a classification tree $T$ 
and its average comparison count in general. 
Let $F^T_{\text{p}}$ and $F^T_{\text{q}}$ be two random variables that denote the
number of elements that are compared
with the smaller and larger pivot first, respectively,  when using $T$. 
Then let $f^{\text{q}}_{s,\ell} = \E\left(F^T_{\text{q}} \mid s,\ell\right)$ resp. 
$f^{\text{p}}_{s,\ell} = \E\left(F^T_{\text{p}} \mid s,\ell\right)$ denote the average number of
comparisons with the larger resp. smaller pivot first, given $s$ and $\ell$. 
Now, if it was decided in each step by independent random experiments
with the correct expectations $s/(n-2)$, $m/(n-2)$, and $\ell/(n-2)$, resp.,
whether an element is small, medium, or large, it would be clear that for example $f^{\text{q}}_{s,\ell} \cdot s /
(n-2)$ is the average number of small elements that are compared with the larger pivot
first. The next lemma shows that one can indeed use this intuition in the
calculation of the average comparison count, except that one gets an additional $O(n^{1 - \varepsilon})$ term due to 
the elements tested not being independent. 

\begin{lemma}\label{lemma:01}
Let $T$ be a classification tree. Let $\E(P^T_n)$ be the average number of 
key comparisons for classifying an input of $n$ elements using $T$. 
Then there exists a constant $\varepsilon > 0$ such that 
\begin{rm}
\begin{align*}
    \E(P^T_n) &= \frac{4}{3}n  +  \frac{1}{\binom{n}{2} 
    \cdot (n-2)} \sum_{s+\ell \leq n-2}\left(f^{\text{q}}_{s,\ell} \cdot s + f^{\text{p}}_{s,\ell} \cdot
\ell\right)+ O\left(n^{1-\varepsilon}\right).
\end{align*}\end{rm}%
\end{lemma}
\begin{proof}
    Fix $p$ and $q$ (and thus $s, m,$ and $\ell$). We will show that 
    \begin{align}\label{eq:20000}
     \E(S^T_2+ L^T_2 \mid s, \ell) = \frac{f^{\text{q}}_{s,\ell} \cdot s + f^{\text{p}}_{s,\ell} \cdot
\ell}{n-2} + O\left(n^{1-\varepsilon}\right).
    \end{align}
    (The lemma then follows by substituting this into \eqref{eq:30}.)

    We call a node $v$ in $T$ \emph{on-track} (to the expected values)  if 
    \begin{align}
        l(v) = \text{q} \text{ and } \bigg|\frac{s}{n-2} - \frac{s - s_v}{n - \lv{v} - 2}\bigg| &\leq
        \frac{1}{n^{1/12}} \quad\quad\quad\text{or} \notag\\
        l(v) = \text{p} \text{ and } \bigg|\frac{\ell}{n-2} - \frac{\ell - \ell_v}{n - \lv{v} - 2}\bigg|
        &\leq \frac{1}{n^{1/12}}.\label{eq:v:good}
    \end{align}
    Otherwise we call $v$ \emph{off-track}.

    We first obtain an upper bound. Starting from \eqref{eq:35}, we
    calculate:
    \begin{align}
    \E(S_2^T + L_2^T \mid s, \ell) &=\sum_{v \in T, l(v) = \text{q}}
	p^v_{s,\ell}\cdot \frac{s - s_v}{n
	- 2 - \lv{v}}+ \sum_{v \in T, l(v) = \text{p}} p^v_{s,\ell} \cdot \frac{\ell -
	\ell_v}{n - 2 -
    \lv{v}}\notag\\
    &=\sum_{v \in T, l(v) = \text{q}}
	p^v_{s,\ell}\cdot\frac{s}{n
	- 2}+ \sum_{v \in T, l(v) = \text{p}} p^v_{s,\ell}
        \cdot\frac{\ell}{n - 2} {}+{}\notag\\ &\quad\quad\sum_{v \in T, l(v) = \text{q}}
	p^v_{s,\ell}\left( \frac{s - s_v}{n
	- 2 - \lv{v}} - \frac{s}{n-2} \right) + \notag\\ &\quad\quad\sum_{v \in T, l(v) = \text{p}} p^v_{s,\ell}
	\left( \frac{\ell -
    \ell_v}{n - 2 -
    \lv{v}}-\frac{\ell}{n-2}\right) \notag\\
	&\leq \sum_{v \in T, l(v) = \text{q}}
	p^v_{s,\ell}\cdot\frac{s}{n
	- 2}+ \sum_{v \in T, l(v) = \text{p}} p^v_{s,\ell}\cdot
        \frac{\ell}{n - 2}{}+{} \notag\\&\quad\quad\sum_{\substack{v \in T, l(v) = \text{q}\\v \text{ on-track}}}
    \frac{p^v_{s,\ell}}{n^{1/12}} + 
	\sum_{\substack{v \in T, l(v) = \text{q}\\v \text{ off-track}}} p^v_{s,\ell} + \sum_{\substack{v \in T, l(v) = \text{p}\\v \text{ on-track}}} 
    \frac{p^v_{s,
    \ell}}{n^{1/12}} +\!\!\! 
    \sum_{\substack{v \in T, l(v) = \text{p}\\v \text{ off-track}}} \!\!p^v_{s,\ell}, \label{eq:act:a}
\end{align}
where the last step follows by separating on-track and off-track nodes and using
\eqref{eq:v:good}.  (For off-track nodes we use that the left-hand side of the
inequalities in \eqref{eq:v:good} is at most $1$.) For the sums in the last line
of \eqref{eq:act:a}, consider each level of the classification tree separately.
Since the probabilities $p^v_{s,\ell}$ for nodes $v$ on the same level sum up to
$1$, the contribution of the $1/n^{1/12}$ terms is bounded by $O(n^{11/12})$. Using
the definition of $f^\text{q}_{s, \ell}$ and $f^\text{p}_{s,\ell}$, we continue
as follows:

\begin{align}\label{eq:10000}
 \E(S_2^T {+} L_2^T \mid s, \ell) &\leq \sum_{v \in T, l(v) = \text{q}}
	p^v_{s,\ell}\cdot\frac{s}{n
	- 2}+ \sum_{v \in T, l(v) = \text{p}} p^v_{s,\ell}\cdot
        \frac{\ell}{n - 2} + \!\!\!\!\!  
	\sum_{\substack{v \in T\\v \text{ off-track}}} p^v_{s,\ell} + O\left(n^{11/12}\right) \notag\\
    &= \frac{f^{\text{q}}_{s,\ell} \cdot s + f^{\text{p}}_{s,\ell} \cdot
\ell}{n-2}  + \sum_{v \in T,
	v \text{ off-track}} p^v_{s, \ell} + O\left(n^{11/12}\right)\notag\\
 &= \frac{f^{\text{q}}_{s,\ell} \cdot s + f^{\text{p}}_{s,\ell} \cdot
\ell}{n-2} + \notag \\ &\,\,\quad\sum_{k = 0}^{n-3} \Pr(\text{the node reached on level $k$ is off-track}) + O\left(n^{11/12}\right),
\end{align}
where in the last step we just rewrote the sum to consider each level in the classification
tree separately.
So, to show \eqref{eq:20000} it  remains to
bound the sum in \eqref{eq:10000} by $O(n^{1-\varepsilon})$. 

To do this, consider a random input that is classified using $T$.
Using an appropriate tail bound, viz. \emph{the method of average bounded differences},  
we will show that with very high probability we do not reach an off-track node in the
classification tree in the first $n - n^{3/4}$ levels. Intuitively, this means 
that it is highly improbable that underway the observed fraction of small elements 
deviates very far from the average $s/(n-2)$. 

Let $X_j$ be the $0$-$1$ random variable that is
$1$
if the $j$-th classified element is small; let $Y_j$ be the $0$-$1$ random
variable that is $1$ if the $j$-th classified element is large. Let $s_i =
X_1 + \dots + X_i$ and $\ell_i = Y_1 + \dots + Y_i$ for $1 \leq i \leq n - 2$.

\begin{claim}
    Let $1 \leq i \leq n - 2$. Then 
    \begin{align*}
	\Pr(|s_i - \E(s_i)| &> n^{2/3}) \leq 2\exp(-n^{1/3}/2), \text{ and}\\
	\Pr(|\ell_i - \E(\ell_i)| &> n^{2/3}) \leq 2\exp(-n^{1/3}/2).
    \end{align*}
\end{claim}

\begin{proof}
    We prove the first inequality. 
    First, we bound the difference $c_j$ between the expectation of
    $s_i$ conditioned on $X_1,\ldots,X_j$ resp. $X_1,\ldots,X_{j-1}$ for $1 \leq
    j \leq i$.
    Using linearity of expectation we may calculate 
    \begin{align*}
	c_j &= \bigl|\E(s_i \mid X_1,\ldots,X_j) - \E(s_i \mid X_1, \ldots,X_{j-1})\bigr| \\
	    &= \bigg|\sum_{k = 1}^j X_k + \sum_{k = j+1}^{i} \frac{s - s_{j}}{n - j - 2} - \sum_{k = 1}^{j-1} X_k - \sum_{k = j}^{i}
	\frac{s - s_{j-1}}{n - j  - 1}\bigg|\\
	    &= \bigg|X_j + \sum_{k = j+1}^{i} \frac{s - s_{j - 1} - X_j}{n - j - 2}  - \sum_{k = j}^{i}
	\frac{s - s_{j-1}}{n - j  - 1}\bigg|\\
	&= \bigg| X_j - X_j \cdot \frac{i - j}{n-j - 2} + (s - s_{j - 1})
	\left(\frac{i-j}{n - j -2} - \frac{i - j + 1}{n - j - 1}\right) \bigg|\\
	&= \bigg| X_j \left(1 - \frac{i - j}{n-j - 2}\right) - (s - s_{j - 1})
	\left(\frac{n - i - 2}{(n - j - 1)(n-j-2)}\right) \bigg|\\
	&\leq \max\left\{\bigg| X_j \left(1 - \frac{i - j}{n-j - 2}\right)\bigg |, \bigg| \frac{s - s_{j
- 1}}
{n - j - 1} \bigg|\right\} \leq 1.
    \end{align*}
    In this situation we may apply the bound known as the method of averaged bounded differences (see
    \cite[Theorem 5.3]{dp09}), which reads
    \begin{align*}
	\Pr(|s_i - \E(s_i)| > t) \leq 2\exp\left(-\frac{t^2}{2\sum_{j \leq
	i}c_j^2}\right),
    \end{align*}
    and get 
    \begin{align*}
        \Pr(|s_i - \E(s_i)| > n^{2/3}) \leq 2\exp\left(\frac{-n^{4/3}}{2i}\right),
    \end{align*}
    which is not larger than $2\exp(-n^{1/3}/2)$.\qed
\end{proof}

Assume that $|s_i - \E(s_i)| = |s_i - i \cdot s/(n-2)| \leq n^{2/3}$. We
have
\begin{align*}
    \left|\frac{s}{n-2} - \frac{s - s_i}{n - 2 - i}\right| &\leq
\left|\frac{s}{n-2} - \frac{s(1 - i/(n-2))}{n-2-i}\right| + \left|\frac{n^{2/3}}{n - 2 -
i}\right|
= \frac{n^{2/3}}{n - 2 - i}.
\end{align*}
That means that for each of the first $i \leq n - n^{3/4}$ levels with very high
probability we are in an \emph{on-track node} on level $i$, because the deviation from the ideal case
that we see a small element with probability $s/(n-2)$ is $n^{2/3} / (n - 2 - i) \leq n^{2/3}/n^{3/4} = 1/n^{1/12}$. 
Thus, for the first $n - n^{3/4}$
levels the contribution of the sums
of the probabilities of off-track nodes in \eqref{eq:10000} is at most $n^{11/12}$. For the last
$n^{3/4}$ levels of the tree, we use that the contribution of the  probabilities
that we reach an off-track node on level $i$ is at most $1$ for a fixed level. 

This shows that the contribution of the sum in \eqref{eq:10000} is $O(n^{11/12})$. This 
finishes the proof of the upper bound on $\E\left(S^T_2+ L^T_2 \mid s, \ell\right)$ given in \eqref{eq:10000}.
The calculations for the lower bound are similar and are omitted here. \qed
\end{proof}
There is the following technical complication when using this lemma in analyzing a strategy 
that is turned into a dual-pivot quicksort algorithm:
The cost bound is $a \cdot n + O(n^{1-\varepsilon})$, and  Hennequin's result (Equation \eqref{eq:10}) cannot be applied directly to such
partitioning costs. 
However, the next theorem says that the leading term of \eqref{eq:10} applies
to this situation as well, and the additional $O(n^{1-\varepsilon})$ term 
in the partitioning cost
is completely covered in the $O(n)$ error term of \eqref{eq:10}.
\begin{theorem}\label{thm:10}
Let $\mathcal{A}$ be a dual-pivot quicksort algorithm that gives rise to
a classification tree $T_n$ for each subarray of length $n$. Assume
$\E(P^{T_n}_n) = a \cdot n + O(n^{1-\varepsilon})$ for all $n$, for some constants $a$ and $\varepsilon > 0$. 
Then $\E\left(C^\mathcal{A}_n\right) = \frac65a n \ln n + O(n)$.
\end{theorem}
\begin{proof}
By linearity of expectation we may split the partitioning cost into two terms $t_1(n) = a \cdot n$ and 
$t_2(n) = K \cdot n^{1-\varepsilon}$, solve recursion \eqref{eq:25} independently for these two
cost terms, and add the results. Applying \eqref{eq:10} for average partitioning cost $t_1(n)$ yields 
an average comparison count of $\frac65 a n \ln n + O(n)$. Obtaining the bound of $O(n)$ for the 
term $t_2(n)$  is a standard application
of the Continuous Master Theorem of Roura \cite{Roura01}, and has been derived for the dual-pivot
quicksort recurrence by Wild, Nebel, and Martínez in a recent technical report
\cite[Appendix D]{WildNM14}. For completeness, 
the calculation is given in Appendix~\ref{app:proof:thm:10}.
\end{proof}
Lemma~\ref{lemma:01} and Theorem~\ref{thm:10} tell us that for the analysis of the 
average comparison count of a dual-pivot quicksort algorithm we just have to find out what 
$f^{\text{p}}_{s,\ell}$ and $f^{\text{q}}_{s,\ell}$ are for
this algorithm. Moreover, to design a good algorithm (w.r.t. the average
comparison count), we should try to make $f^{\text{q}}_{s,\ell} \cdot s + 
f^{\text{p}}_{s,\ell} \cdot
\ell$ small for each pair $s,\ell$. 

\section{Analysis of Some Known Classification Strategies}\label{sec:methods}
In this section, we will study different classification strategies in the light
of the formulas from
Section~\ref{sec:additional:cost:term}.  

\paragraph{Oblivious Strategies}
We will first consider strategies that do not use information of previous
classifications for future classifications. To this end, we call a classification tree
\emph{oblivious} if for each level all nodes $v$ on this level share the same
label $l(v) \in \{\text{p},\text{q}\}$. 
This means that
these algorithms do not react
to the outcome of previous classifications, but use a fixed sequence of pivot
choices.
Examples for
such strategies are, e.\,g., 
\begin{itemize}
    \item always compare to the smaller pivot first, 
    \item always
compare to the larger pivot first,
    \item alternate the pivots in each step.
\end{itemize}
Let $f^\text{q}_n$ denote the average number of comparisons to the larger pivot
first. By assumption this value is independent of $s$ and $\ell$. Hence these
strategies make sure that $f^\text{q}_{s,\ell} = f^\text{q}_n$ and
$f^\text{p}_{s,\ell} = n - 2 - f^\text{q}_n$ for all pairs of values $s, \ell$.

Applying Lemma~\ref{lemma:01} gives us 
\begin{align*}
    \E(P_n) &=
    \frac{4}{3}n + \frac{1}{\binom{n}{2} \cdot (n - 2)} \cdot \sum_{s + \ell
    \leq n - 2} \left(f^\text{q}_n \cdot s + (n - 2 - f^\text{q}_n) \cdot \ell\right) + O(n^{1 - \varepsilon})\\
    &= \frac43n + \frac{f^\text{q}_n}{\binom{n}{2} \cdot (n - 2)} \cdot \left(\sum_{s + \ell \leq n - 2} s \right)+ 
    \frac{n - 2- f^\text{q}_n}{\binom{n}{2} \cdot (n - 2)} \cdot \left(\sum_{s + \ell \leq n - 2} \ell \right) + O(n^{1 - \varepsilon})\\
    &= \frac43n + \frac{1}{\binom{n}{2}} \cdot \left(\sum_{s + \ell \leq n - 2} s \right) + O(n^{1 - \varepsilon}) = \frac53n + O(n^{1 - \varepsilon}).
\end{align*}
Using Theorem~\ref{thm:10} we get $\E(C_n) = 2  n \ln n +
O(n)$---the leading term being the same as in standard quicksort. 
So, for each strategy that does not adapt to the outcome of previous
classifications, there is no difference to the average
comparison count of classical quicksort. Note that this also holds for
\emph{randomized strategies} such as ``flip a coin to choose the pivot used in the
first comparison'', since such a strategy can be seen as a probability
distribution on oblivious strategies.

\paragraph{Yaroslavskiy's Algorithm}
Following \cite[Section 3.2]{nebel12}, Yaroslavskiy's algorithm is an
implementation of the following
strategy $\mathcal{Y}$: \emph{Compare
 $\ell$ elements to $q$ first, and compare the other elements to $p$ first.}
We get that $f^\text{q}_{s,\ell} = \ell$ and
$f^\text{p}_{s,\ell} = s + m$. 
Applying Lemma~\ref{lemma:01}, we get 
\begin{align*}
\E(P^\mathcal{Y}_n) &= \frac{4}{3} n + \frac{1}{\binom{n}{2}}
\sum_{s + \ell \leq n - 2} \left(\frac{s\ell}{n-2} +
    \frac{(s+m)\ell}{n-2}\right) + O(n^{1 - \varepsilon}). 
\end{align*} 
Of course, it is possible to evaluate this sum by hand. We used Maple\textsuperscript{\textregistered} to obtain 
$\E(P^\mathcal{Y}_n) =  \frac{19}{12} n + O(n^{1 - \varepsilon})$.
Using Theorem~\ref{thm:10} gives  
$\E(C_n) = 1.9 n\ln n + O(n)$, as in \cite{nebel12}.

\paragraph{Sedgewick's Algorithm}\label{sec:sedgewick}
Following \cite[Section 3.2]{nebel12}, Sedgewick's algorithm amounts to an
implementation of the following strategy
$\mathcal{S}$: \emph{Compare (on average) a fraction of
$s/(s+\ell)$ of the keys with $q$ first, and compare the other keys with
$p$ first.}
We get $f^\text{q}_{s,\ell} = (n-2) \cdot
s/(s+\ell)$ and $f^\text{p}_{s,\ell} = (n-2) \cdot
\ell / (s+\ell)$.

Plugging these values into Lemma~\ref{lemma:01}, we calculate 
\begin{align*}
    \E(P^\mathcal{S}_n) &= \frac{4}{3} n + \frac{1}{\binom{n}{2}}
            \sum_{s + \ell \leq n - 2} \left(\frac{s^2}{s + \ell} +
            \frac{\ell^2}{s + \ell}\right) + O(n^{1 - \varepsilon})  =
            \frac{16}{9} n +  O(n^{1 - \varepsilon}).
\end{align*}

Applying Theorem~\ref{thm:10} gives $\E(C_n) = 2.133...
\cdot n \ln n +
O(n)$, as known from \cite{nebel12}.

Obviously, this is worse than the oblivious strategies considered before.\footnote{We remark that in his thesis
Sedgewick \cite{sedgewick} focused on the average number of swaps, not on the
comparison count.}   This is
easily explained intuitively: If the fraction of small elements is large, it
will compare many elements with $q$ first. But this costs two comparisons for
each small element. Conversely, if the fraction of large elements is large, it
will compare many elements to $p$ first, which is again the wrong decision.

Since Sedgewick's strategy seems to do exactly the opposite of what one should
do to lower the comparison count, we consider
the following modified 
strategy $\mathcal{S}'$: \emph{For given $p$
and $q$, compare (on average) a fraction of
$s/(s+\ell)$ of the keys with $p$ first, and compare the other keys with
$q$ first.} ($\mathcal{S}'$ simply uses $p$ first when $\mathcal{S}$ would use
$q$ first and vice versa.)

Using the same analysis as above, we get $\E(P_n) =
\frac{14}{9} n
+ O(n^{1 - \varepsilon})$, which yields $\E(C_n) = 1.866...\cdot n
\ln n + O(n)$---improving on the standard algorithm and even on
Yaroslavskiy's algorithm! Note that this has been observed by Wild in his Master's Thesis as well \cite{Wild2013}.

\paragraph{Remark} Swapping the first comparison with $p$ and $q$ as in the strategy described above is a
general technique. In fact, if the
leading coefficient of the average number of comparisons for a fixed rule 
for choosing $p$ or $q$ first is $\alpha$, e.\,g., $\alpha = 2.133...$ for strategy
$\mathcal{S}$,  then the leading coefficient of the strategy that
does the opposite is $4 - \alpha$, e.\,g., $4 - 2.133... = 1.866...$ as in
strategy $\mathcal{S}'$.

\section{An Asymptotically Optimal Classification Strategy}\label{sec:decreasing}
Looking at the previous sections, all strategies used the idea that we should
compare a certain fraction of elements to $p$ first, and the other
elements to $q$ first. In this section, we
will study the following strategy $\mathcal{I}$: \emph{If $s >
\ell$ then always compare with $p$ first, otherwise always compare with $q$
first.} 

Of course, for an implementation of this strategy we have to deal with the
problem of finding out which case applies before the comparisons have been made.
We shall analyze a guessing strategy to resolve this. 

    \subsection{Analysis of the Idealized Classification Strategy}\label{sec:optimal:strategy}
 Assume for a moment that for a given random input with pivots $p,q$ the
 strategy ``magically'' knows whether
 $s > \ell$ or not and correctly determines the pivot that should be used for all
 comparisons. For fixed $s$ and $\ell$ this means that for $s > \ell$ the
 classification strategy makes exactly $\ell$ additional comparisons, and for $s \leq \ell$ it
 makes $s$ additional comparisons.

 When we start from \eqref{eq:30}, 
 a standard calculation shows that for this strategy
 \begin{align}\label{eq:80}
     \E(P_n) &= \frac{4}{3} n + \frac{1}{\binom{n}{2}}\sum_{s +
 \ell \leq n - 2} \min(s, \ell)  = \frac{3}{2} n + O(1).
 \end{align}
 Applying \eqref{eq:10}, we get
$\E(C_n) = 1.8 n \ln n + O(n)$,
which is by $0.1n\ln n$ smaller than the average number of key
comparisons in Yaroslavskiy's algorithm.

To see that this method is asymptotically optimal, 
recall that according to
Lemma~\ref{lemma:01} 
the average comparison count is determined up to a linear term by
the parameters $f^\text{q}_{s,\ell}$ and
$f^\text{p}_{s,\ell} = n - 2 - f^\text{q}_{s,\ell}$.
Strategy $\mathcal{I}$ chooses these values such that $f^\text{q}_{s,\ell}$ is
either $0$ or $n-2$, minimizing each term of the sum in
Lemma~\ref{lemma:01}---and thus minimizing the sum.

\subsection{Guessing Whether $s < \ell$ or not}
We explain how the idealized classification strategy just described can be approximated by an
implementation. The idea is to make a few comparisons and use the outcome
as a basis for a guess.

After $p$ and $q$ are chosen,  classify the first $\Samplesize=O(n^{1-\varepsilon})$ many
elements (the \emph{sample})
and calculate $s'$ and $\ell'$, the number of
small and large elements in the sample. If $s' < \ell'$, compare the remaining
$n - 2 - \Samplesize$ elements with $q$ first, otherwise  compare them with
$p$ first. We say that the guess was \emph{correct} if 
$s' < \ell'$ and $s < \ell$ or $s' \geq \ell'$ and $s \geq \ell$. In order not to clutter up formulas,
we will always assume that $n^{1-\varepsilon}$ is an integer. One would otherwise
work with $\lceil n^{1-\varepsilon}\rceil$. 

We incorporate guessing errors and sampling cost into \eqref{eq:80} as follows:
\begin{align}\label{eq:500}
    \E(P_n) &= \frac{4}{3} n + 
    \frac{1}{\binom{n}{2}} \sum_{s + \ell \leq n
    - 2}\bigg(\Pr\left(\text{guess
    correct} \mid s, \ell\right) \cdot \min(s, \ell){} +{}\notag \\[-1em]&\quad\quad\quad\quad\quad\quad\quad\quad\quad\quad\,\,\Pr(\text{guess wrong} \mid s, \ell) \cdot
	\max(s, \ell)\bigg) + O(n^{1 - \varepsilon})\notag\\[-1em]
    &=  \frac{4}{3}n + \frac{2}{\binom{n}{2}}\sum_{s=0}^{n/2}\sum_{\ell =
s+1}^{n-s} \bigg(\Pr(\text{guess correct} \mid s, \ell) \cdot
s{} +{} \notag\\[-1em]
&\quad\quad\quad\quad\quad\quad\quad\quad\quad\quad\,\,\,\,\,\Pr(\text{guess
wrong} \mid s, \ell) \cdot \ell\bigg) + O(n^{1 - \varepsilon}),
\end{align}
% Our estimation is correct iff $\texttt{d} < 0$. In
%general, we can calculate this probability as
    %     \begin{align*}
    % \Pr(\texttt{d} < 0 \mid s \text{ small}, \ell \text{ large elements}) = \sum_{s' = 0}^{\lfloor
    % (\texttt{samplesize}-1)/2\rfloor}\sum_{\ell' = s' +
    % 1}^{\texttt{samplesize} - \ell'}\frac{\binom{s}{s'}\binom{\ell}{\ell'}\binom{n-\ell-s}{\texttt{samplesize} -
    %     s' - \ell'}}{\binom{n}{\texttt{samplesize}}}.
    %     \end{align*}
%
where the cost for comparing the elements in the sample is covered by the $O(n^{1 - \varepsilon})$ term.
The following lemma says that for a wide range of values $s$ and $\ell$ the
probability of a guessing error is exponentially small. 
%, i.\,e., we can ignore
%the contributions of wrong guesses in our calculations.
\begin{lemma}\label{lem:estimation:error}
    Let $s$ and $\ell$ with $s \leq \ell - n^{3/4}$ and $\ell \geq n^{3/4}$ for
    $n \in \mathbb{N}$ be given.
    Let $\textnormal{\Samplesize} = n^{2/3}$. Then 
    $\Pr(\textnormal{guess wrong} \mid s, \ell) \leq \exp\left(-n^{1/6}/18\right).$
\end{lemma}
\begin{proof}
    Let $X_i$ be a random variable that is $-1$ if
    the $i$-th classified element of the sample is large, $0$ if it is medium, 
    and $1$ if it is small. Let $d = \sum_{i =
    1}^{\Samplesize} X_i$. 

    As in the proof of Lemma~\ref{lemma:01}, we want to apply the method of
    averaged bounded differences. Using the assumptions on the values of $s$ and $\ell$, straightforward
    calculations show that $\E(d) \leq -\Samplesize/n^{1/4} = -n^{5/12}$. Furthermore, we have that 
    \begin{align*}
	c_i = \bigl|\E(d \mid X_1,\ldots,X_i) - \E(d \mid X_1, \ldots, X_{i-1})\bigr| \leq
    3, \text{ for  } 1\leq i \leq \Samplesize.
    \end{align*}
    To see this, we let $s_i$ and $\ell_i$ denote the number of small and
    large  elements, respectively, that
    are still present after the first $i$ elements have been classified, i.\,e., $X_1,\ldots,X_i$ have been determined. Let $Y_i$ be the
    $0$-$1$ random variable that is $1$ if and only if $X_i$ is $1$, and let $Z_i$ be the
    $0$-$1$ random variable that is $1$ if and only if $X_i$ is $-1$. 
    We calculate:
    \begin{align*}
        |&\E(d \mid X_1, \ldots, X_i) - \E(d \mid X_i,\ldots, X_{i-1})|\\
    	&= \Bigg| \sum_{j = 1}^i{X_j} + \sum_{j = i+1}^{\Samplesize} \bigg[\Pr(X_j = 1 \mid X_1,
    \ldots, X_i) - \Pr(X_j = -1 \mid X_1,\ldots, X_i)\bigg] \\ &\quad\quad- \sum_{j =
    1}^{i-1}{X_j}- \sum_{j = i}^{\Samplesize} \bigg[\Pr(X_j = 1 \mid X_1,
    \ldots, X_{i-1}) - \Pr(X_j = -1 \mid X_1,\ldots, X_{i-1})\bigg]\Bigg|\\
    &= \Bigg| X_i + \sum_{j = i + 1}^{\Samplesize} \left[\frac{s_i}{n - i} -
    \frac{\ell_i}{n-i}\right] - \sum_{j = i}^{\Samplesize} \left[\frac{s_{i - 1}}{n-i + 1} -
    \frac{\ell_{i - 1}}{n-i+1}\right]\Bigg|\\
    &= \Bigg| X_i + (\Samplesize - i)\cdot  \left[\frac{s_{i-1} - Y_i}{n - i} -
    \frac{\ell_{i-1} - Z_i}{n-i}\right] - (\Samplesize - i +1) \cdot \left[\frac{s_{i - 1}}{n-i + 1} -
    \frac{\ell_{i - 1}}{n-i+1}\right]\Bigg|\\
    &= \Bigg| X_i  + \frac{(\Samplesize-i) \cdot (Z_i - Y_i)}{n-i} + s_{i - 1}
    \left[\frac{\Samplesize-i}{n - i} -
    \frac{\Samplesize-i+1}{n-i+1}\right] + \ell_{i - 1} \left[\frac{\Samplesize - i + 1}{n-i+1} {-}
    \frac{\Samplesize - i}{n-i}\right]\Bigg|\\
    &= \Bigg|X_i  + \frac{(\Samplesize-i) \cdot (Z_i - Y_i)}{n-i} + (\ell_{i -
    1} - s_{i - 1}) \cdot \frac{n-\Samplesize}{(n-i)(n-i+1)}
    \Bigg|\\
    &\leq \left\vert X_i\right\vert + \left\vert Z_i - Y_i \right \vert +  \left\vert\frac{\ell_{i-1} - s_{i - 1}}{n - i + 1}\right\vert \leq 3.
    \end{align*}
    The method of averaged bounded differences (see
    \cite[Theorem 5.3]{dp09}) now implies that 
    \begin{align*}
	\Pr(d > \E(d) + t) \leq \exp\left(-\frac{t^2}{2\sum_{i \leq
	\Samplesize}c_i^2}\right), \text{ for } t > 0,
    \end{align*}
    which with $t = n^{5/12} \leq -\E(d)$ yields
    \begin{align*}
	\Pr(d > 0) \leq \exp\left(-\frac{n^{1/6}}{18}\right).\tag*{\qed}
    \end{align*}
\end{proof}
Of course, we get an analogous result for $s \geq n^{3/4}$ and $\ell \leq s - n^{3/4}$.

%The proof is given by a standard application of the method of averaged bounded
%differences. It can be found in Appendix~\ref{app:sec:lem:estimation:error}.
%Of course, we get an equivalent result for the probability that $\Pr(d < 0)$ in
%the case that $\ell \leq s - n^{3/4}$
%and $s \geq n^{3/4}$.
%
Classification strategy $\mathcal{SP}$ works as follows: Classify the first 
$\Samplesize = n^{2/3}$ elements. Let $s'$ $[\ell']$ be the number of elements classified as being 
small [large].
If $s' > \ell'$ then use $p$ for
the first comparison for the remaining elements, otherwise use $q$.

We can now analyze the average number of key comparisons of this strategy turned
into a dual-pivot quicksort algorithm. 
\begin{theorem}
    The average comparison count
    of strategy $\mathcal{SP}$ turned into a dual-pivot quicksort algorithm is $
    1.8  n\ln n + O(n)$.
\end{theorem}

\begin{proof}
    We only have to analyze the expected classification cost. First, we classify
    {\Samplesize} many elements. The number of key comparisons for
    these classifications is at most
    $2 n^{2/3} = O(n^{1 - \varepsilon})$. By symmetry, we may focus on the case that $s \leq
    \ell$. We distinguish the following three cases:
    \begin{enumerate}
	\item $\ell \leq n^{3/4}$: The
	    contribution of terms in \eqref{eq:500} satisfying this case is at most
		$$\frac{2}{\binom{n}{2}} \sum_{\ell =
		0}^{n^{3/4}}\sum_{s = 0}^{\ell} \ell = O(n^{1/4}).$$
	\item $\ell - n^{3/4} \leq s \leq \ell$:
            Let $\textrm{m}_1(s, n) = \min(s + n^{3/4}, n - s)$. The contribution
            of terms in \eqref{eq:500} satisfying this case is at most
	     $$\frac{2}{\binom{n}{2}} \sum_{s = 0}^{n/2} \,\, \sum_{\ell
             = s}^{\textrm{m}_1(s, n)} \ell = O(n^{3/4}).$$
     \item $\ell \geq n^{3/4}$ and $s \leq \ell - n^{3/4}$. Let $\textrm{m}_2(\ell, n) =
	 \min(n-\ell, \ell - n^{3/4})$.
	    Following Lemma~\ref{lem:estimation:error}, the probability of
	    guessing wrong is at most $\exp(-n^{1/6}/18)$.
	    The contribution of this case in \eqref{eq:500} is hence at most 
	    \begin{align*}
		&\frac{2}{\binom{n}{2}} \sum_{\ell = n^{3/4}}^{n}\sum_{s =
	    0}^{\textrm{m}_2(\ell, n)} \bigg(s + \exp(-n^{1/6}/18) \ell\bigg) =
	    \left(\frac{2}{\binom{n}{2}} \sum_{\ell = n^{3/4}}^{n}\sum_{s =
	    0}^{\textrm{m}_2(\ell, n)} s\right) + o(1) = \frac{n}{6} + O(1).
	     \end{align*}
    \end{enumerate}
    Using these contributions in \eqref{eq:500}, we
    expect a partitioning step to make $\frac{3}{2}
    n + O(n^{1 - \varepsilon})$ key comparisons. Applying Theorem~\ref{thm:10}, we get
    $\E(C_n) = 1.8 n \ln n + O(n)$. \qed
\end{proof}
In this section we have seen an asymptotically optimal strategy. In the next section we will present the
optimal classification strategy. Unfortunately, it is even more unrealistic than the idealized
strategy $\mathcal{I}$ from above. However, we will give an implementation that comes
very close to the optimal strategy in terms of the number of comparisons. 

\section{The Optimal Classification Strategy and its Implementation}\label{sec:optimal:strategies}
We will consider two more strategies, an optimal (but not algorithmic) strategy, and 
an algorithmic strategy that is optimal up to a very small error term.

We first study the
(unrealistic!)
setting where $s$ and $\ell$, i.\,e., the number of small resp. large elements, are
known to the algorithm after the pivots are chosen, and the classification tree can
have different node labels for each such pair of values.
Recall that
$s_v$ and $\ell_v$, resp., denote the number of
elements 
that have been classified as small and large, resp., when at node $v$ in the classification tree.
We consider the following
strategy $\mathcal{O}$: \emph{Given $s$ and $\ell$, the comparison at node $v$ is with
the smaller pivot first if $s - s_v > \ell - \ell_v$, otherwise it is
with the larger pivot
first.}\footnote{This strategy was suggested to us by Thomas Hotz (personal
communication).}

\begin{theorem}\label{thm:O:optimal}
    Strategy $\mathcal{O}$ is optimal, i.\,e., its ACT is at most as large as  ACT$_T$ for
all classification trees $T$. When using $\mathcal{O}$ in a dual-pivot quicksort
algorithm, we have $\E(C^\mathcal{O}_n) = 1.8 n \ln n + O(n)$.
\end{theorem}
\begin{proof}
    The proof of the first statement (optimality) is surprisingly simple. Fix the two pivots, and consider equation
    \eqref{eq:35}. For each node $v$ in the classification tree Strategy $\mathcal{O}$ chooses the label
    that minimizes the contribution of this node to \eqref{eq:35}. So,
    it minimizes each term of the sum, and thus minimizes the additional cost
    term in \eqref{eq:30}. 
    We now prove the second statement. For this, we first derive an upper bound of $1.8 n \ln n + O(n)$ for the 
    average number of comparisons, and then show that this is tight. 

    For the first part, let an input with $n$ entries and two pivots be given,
    so that there are $s$ small and $\ell$ large elements. Assume $s \geq
    \ell$. Omit all medium elements to obtain a reduced input $(a_1, \ldots,
    a_{n'})$ with $n' = s + \ell$. For $0 \leq i \leq n'$ let $s_i$ and
    $\ell_i$ denote the number of small resp. large elements remaining in
    $(a_{i + 1}, \ldots, a_{n'})$. Let $D_i = s_i - \ell_i$. Of course we have
    $D_0 = s - \ell$ and $D_{n'} = 0$.  Let $i_1 < i_2 < \dots  < i_k$ be the
    list of indices $i$ with $D_i = 0$. (In particular, $i_k = n'$.) Rounds $i$
    with $D_i = 0$ are called \emph{zero-crossings}. Consider some $j$ with
    $D_{i_j} = D_{i_{j + 1}} = 0$. The numbers $D_{i_j + 1}, \ldots,  D_{i_{j +
    1} - 1}$ are nonzero and have the same positive [or negative] sign. The
    algorithm compares $a_{i_j + 2},\ldots, a_{i_{j + 1}}$ with the smaller
    [larger] pivot first, and $a_{i_j + 1}$ with the larger pivot first.  Since
    $\{a_{i_j + 1},\ldots, a_{i_{j + 1}}\}$ contains the same number of small
    and large elements, the contribution of this segment to the additional
    comparison count is $\frac12(i_{j + 1} - i_j) - 1$ $\left[\text{or }
    \frac12(i_{j + 1} - i_j)\right]$.

    If $D_0 > 0$, i.\,e., $s > \ell$, all elements in $\{a_1, \ldots,
    a_{i_1}\}$ are compared with the smaller pivot first, and this set contains
    $\frac12 (i_1 - (s - \ell))$ large elements (and $\frac12(i_1 + (s -
    \ell))$ small elements), giving a contribution of $\frac12 (i_1 - (s -
    \ell))$ to the additional comparison count. Overall, the additional
    comparison count $S_2 + L_2$ (see the end of
    Section~\ref{sec:average:case:analysis}, in particular \eqref{eq:30}) of
    strategy $\mathcal{O}$ on the considered input is 

    \begin{align}
        \frac{i_1 - (s - \ell)}{2} + \sum_{j = 1}^{k - 1} \frac{i_{j + 1} - i_j}{2} - k^\ast = \frac{n' - (s - \ell) }{2} - k^\ast = \ell - k^\ast,
        \label{eq:o:act}
    \end{align}
    for some correction term $k^\ast \in \{0,\ldots,k - 1\}$.

    Averaging the upper bound $\ell$ over all pivot choices, we obtain the following bound for the additional cost term of strategy $\mathcal{O}$:
    \begin{align}\label{eq:5000}
        \E\left(S_2 + L_2\right) \leq \frac{1}{\binom{n}{2}} \cdot \left( 2 \cdot \sum_{\substack{s + \ell \leq n \\\ell < s}} \ell + \sum_{\ell \leq n/2} \ell\right).
    \end{align}
    This gives an average number of at most $1.5n + O(1)$ comparisons. For such
    a
    partitioning cost we can use \eqref{eq:10} and obtain an 
    average comparison count for sorting via strategy $\mathcal{O}$ of at most $1.8 n\ln n + O(n)$.

    It remains to show that this is tight. We shall see that the essential step in this analysis
    is to show that the average (over all inputs) of the number of \emph{zero-crossings} 
    (the number $k$ from above, excepting the zero-crossing at position $n$)
    is $O(\log n)$. Again, we temporarily omit medium elements to simplify
    calculations, i.\,e., we assume that the number of small and large elements
    together is $n'$. Fix a position $n' - 2i$, for $1 \leq i \leq n'/2$. If $D_{n' - 2i} = 0$, 
    then there are as many small elements as there are large elements in the last $2i$ positions
    of the input. Consequently, the input has to contain between $i$ and $n' - i$ small elements, otherwise
    no zero-crossing is possible.
    The probability that a random input (excepting the two pivots) of $n'$ elements has exactly $s$ small elements is $1/(n' + 1)$, for $0 \leq s \leq n'$.
    Let $Z_{n'}$ be the random variable that denotes the number of
    zero-crossings for an input of $n'$ elements excepting the two pivots. We calculate:

\begin{align*}
    \E(Z_{n'}) &= \sum_{1 \leq i \leq n'/2} \Pr(\text{there is a zero-crossing at
position $n' - 2i$})\\
&= \frac{1}{n' + 1} \sum_{i = 1}^{n'/2} \sum_{s = i}^{n' - i} \Pr(D_{n' - 2i} = 0 \mid \text{$s$ small elements})\\
&= \frac{1}{n' + 1} \sum_{i = 1}^{n'/2} \sum_{s = i}^{n' - i}
\frac{\binom{2i}{i} \cdot \binom{n' - 2i}{s-i}}{\binom{n'}{s}}
\leq \frac{2}{n' + 1} \sum_{i = 1}^{n'/2} \sum_{s = i}^{n'/2}
\frac{\binom{2i}{i} \cdot \binom{n' - 2i}{s-i}}{\binom{n'}{s}},
\end{align*}
where the last step follows by symmetry: replace $s > n'/2$ by $n'-s$.

By using the well-known estimate $\binom{2i}{i} = 
\Theta(2^{2i}/\sqrt{i})$ (which follows directly from Stirling's
approximation), we
continue by
\begin{align}
    \E(Z_{n'})
    &= \Theta\left(\frac{1}{n'}\right) \sum_{i =1}^{n'/2} \frac{2^{2i}}{\sqrt{i}}
\sum_{s=i}^{n'/2} 
    \frac{\binom{n'-2i}{s - i}}{\binom{n'}{s}}\notag\\
&=\Theta\left(\frac{1}{n'}\right) \sum_{i=1}^{n'/2} \frac{2^{2i}}{\sqrt{i}} \sum_{s =
i}^{n'/2} 
\frac{(n' - s) \cdot \ldots \cdot (n' - s - i + 1) \cdot s 
\cdot \ldots \cdot (s-i + 1)}{n' \cdot \ldots \cdot (n' - 2i + 1)}\notag\\
&=\Theta\left(\frac{1}{n'}\right) \sum_{i = 1}^{n'/2} \frac{n'+1}{\sqrt{i}(n' {-} 2i {+} 1)} \sum_{j = 0}^{n'/2-i}
\prod_{k = 0}^{i - 1} \frac{(n'+2j - 2k)(n'-2j - 2k)}{(n' - 2k + 1)(n' -
2k)},\label{eq:zero:crossings:0}
\end{align}
where the last step follows by an index transformation using $j = n'/2 - s$ and multiplying $2^{2i}$ into the 
terms of the rightmost fraction. We now obtain an upper bound for the rightmost product:

\begin{align*}
\prod_{k = 0}^{i - 1} \frac{(n'+2j - 2k)(n'-2j - 2k)}{(n' - 2k + 1)(n' -
2k)} \leq \prod_{k = 0}^{i - 1} \left(1 - \left(\frac{2j}{n'-2k}\right)^2\right) 
\leq \left(1 - \left(\frac{2j}{n'}\right)^2\right)^i.
\end{align*}
We substitute this bound into \eqref{eq:zero:crossings:0}. Bounding the rightmost sum of
\eqref{eq:zero:crossings:0} by an integral, and using Maple\textsuperscript{\textregistered},
we obtain
\begin{align*}
\E(Z_{n'}) &=  O\left(\frac{1}{n'}\right) \sum_{i=1}^{n'/2} \frac{n'+1}{\sqrt{i}(n' - 2i + 1)} \left(\int_0^{n'/2} 
    \left(1 - \left(\frac{2t}{n'}\right)^2\right)^{i} \text{ $dt$} + 1\right)\\
    &=O\left(\frac{1}{n'}\right) \sum_{1 \leq i \leq n'/2} \frac{n'+1}{\sqrt{i}(n' - 2i + 1)} 
        \cdot \left(n' \cdot \frac{\Gamma(i+1)}{\Gamma(i + 3/2)} + 1 \right),\\
\end{align*}
    involving the Gamma function $\Gamma(x) =
    \int_{0}^{\infty}t^{x-1}e^{-t} \text{ $dt$}$.
    Since $\Gamma(i+1)/\Gamma(i+3/2) = \Theta(1/\sqrt{i})$, we may continue by
calculating

\begin{align*}
    \E(Z_{n'}) &= O\left(\sum_{i = 1}^{n'/2} \frac{n'+1}{i(n' - 2i + 1)}\right) = O\left(\sum_{i = 1}^{n'/4} \frac{1}{i} + 
    \sum_{i = n'/4 + 1}^{n'/2}\frac{1}{n' - 2i  + 1}\right) = O(\log n').
\end{align*}
Now we generalize the analysis to the case that the input contains medium
elements. Let $(a_1, \ldots, a_n)$ be a random input. Omit all
medium elements to obtain a reduced input $(a_1, \ldots, a_{n'})$ with $n' = s
+ \ell$.  The additional cost term of strategy $\mathcal{O}$ on the reduced
input is the same as on the original input, since medium elements
influence neither the decisions of the strategy on elements of the reduced input 
nor the additional cost term.
Starting from \eqref{eq:o:act}, we may bound the difference between the additional cost 
term $\E(S_2 + L_2)$ and the bound given in \eqref{eq:5000} by the 
average over all $s, \ell$ of the values $k^\ast$. This is bounded by the average over all $s, \ell$
of the values $Z_{n'}$, hence by
\begin{align*}
    \frac{1}{\binom{n}{2}} \left( 2 \sum_{\substack{s + \ell \leq n \\\ell < s}} \ell + \sum_{\ell \leq n/2} \ell\right) - \E(S_2 + L_2) & \leq \frac{1}{\binom{n}{2}} \sum_{s + \ell \leq n - 2} \E(Z_{s + \ell} \mid
    s, \ell) \\ &= \frac{1}{\binom{n}{2}} \sum_{s + \ell \leq n - 2} O(\log (s+\ell)) = 
    O(\log n).
\end{align*}
Since $O(\log n)$ is in $O(n^{1-\varepsilon})$, the influence of these $O(\log n)$ terms
to the total average sorting cost is bounded by $O(n)$, see Theorem~\ref{thm:10} and
Appendix~\ref{app:proof:thm:10}.\qed
\end{proof}
While making an algorithm with minimum ACT possible, the assumption that the exact
number of small and large elements is known is of course not true for a real
algorithm or for a fixed tree. We can, however, identify a real, algorithmic
partitioning strategy whose ACT differs from the optimal one only by a
logarithmic term.  We study the following strategy $\mathcal{C}$: \emph{The
comparison at node $v$ is with the smaller pivot first if $s_v > l_v$,
otherwise it is with the larger pivot first.}

While $\mathcal{O}$ looks into the future (``Are there more small elements or
more large elements left?''), strategy $\mathcal{C}$ looks into the past
(``Have I seen more small or more large elements so far?'').
It is not hard to see that for some inputs the number of additional
comparisons of
strategy $\mathcal{O}$ and $\mathcal{C}$ can differ significantly. The next
theorem shows that averaged over all possible inputs, however, there is only a
small
difference.
\begin{theorem}\label{thm:optimal:strategies:diff}
    Let ACT$_\mathcal{O}$ resp. ACT$_\mathcal{C}$ be the ACT for 
    classifying $n$ elements using strategy $\mathcal{O}$ resp.
    $\mathcal{C}$. Then
    ACT$_\mathcal{C} = \text{ACT}_\mathcal{O} + O(\log n)$.
   When using $\mathcal{C}$ in a dual-pivot quicksort algorithm, we get
$\E(C^\mathcal{C}_n) = 1.8 n \ln n + O(n)$.
\end{theorem}

\begin{proof}
Assume that strategy
$\mathcal{O}$ inspects the elements in the order $a_{n-1},\ldots,a_2$, while
$\mathcal{C}$ uses the order $a_2,\ldots,a_{n-1}$. If the strategies compare
the element $a_{i}$ to different pivots, then there are
exactly as many small elements as there are large elements in $\{a_2,\ldots,a_{i-1}\}$ or $\{a_2,\ldots,a_i\}$, depending
on whether $i$ is even or odd, see Figure~\ref{fig:strategy:difference}.

    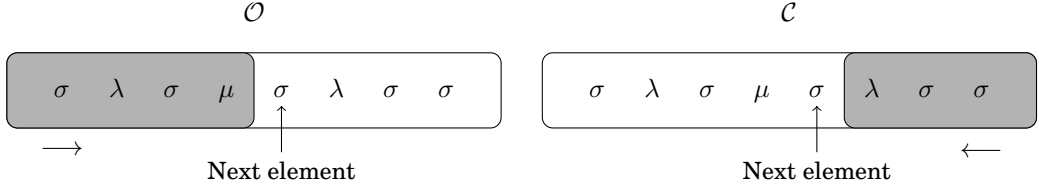
\begin{figure}[t!]
    \centering
	    \scalebox{1}{
		\begin{tikzpicture}[xscale=0.73]
	    \draw[draw, rounded corners] (0,0) rectangle (9,1);
	    \draw[draw, rounded corners, fill=black!30] (0,0) rectangle (4.5,1);
	    \node (label) at (1, -0.3) {$\longrightarrow$};
        \node[anchor=mid] (l3) at (1,0.5) {$\sigma$};
        \node[anchor=mid] (l2) at (2,0.5) {$\lambda$};
	    \node[anchor=mid] (l8) at (3,0.5) {$\sigma$};
	    \node[anchor=mid] (l5) at (4,0.5) {$\mu$};
	    \node[anchor=mid] (l1) at (5,0.5) {$\sigma$};
	    \node[anchor=mid] (l4) at (6,0.5) {$\lambda$};
	    \node[anchor=mid] (l7) at (7,0.5) {$\sigma$};
	    \node[anchor=mid] (l6) at (8,0.5) {$\sigma$};

	    \node[anchor=mid] (label) at (5, -0.6) {\small{Next element}};
	    \draw[->] (label) -- (l1);

            \node[anchor=mid] (o1) at (4.5, 1.5) {$\mathcal{O}$};

	    \begin{scope}[xshift=9.75cm]
	    \draw[draw, rounded corners] (0,0) rectangle (9,1);
	    \draw[draw, rounded corners, fill=black!30] (5.5,0) rectangle (9,1);
	    \node[anchor=mid] (label1) at (8, -0.3) {$\longleftarrow$};
	    \node[anchor=mid] (l3a) at (1,0.5) {$\sigma$};
	    \node[anchor=mid] (l2a) at (2,0.5) {$\lambda$};
	    \node[anchor=mid] (l8a) at (3,0.5) {$\sigma$};
	    \node[anchor=mid] (l5a) at (4,0.5) {$\mu$};
	    \node[anchor=mid] (l1a) at (5,0.5) {$\sigma$};
	    \node[anchor=mid] (l4a) at (6,0.5) {$\lambda$};
	    \node[anchor=mid] (l7a) at (7,0.5) {$\sigma$};
	    \node[anchor=mid] (l6a) at (8,0.5) {$\sigma$};

        \node[anchor=mid] (labela) at (5, -0.6) {\small{Next element}};
	    \draw[->] (labela) -- (l1a);
            
        \node[anchor=mid] (o2) at (4.5, 1.5) {$\mathcal{C}$};
	\end{scope}
\end{tikzpicture}}
    \caption{Visualization of the decision process when inspecting an element using 
    strategy $\mathcal{O}$ (left) and $\mathcal{C}$ (right). Applying strategy $\mathcal{O}$
    from left to right uses that of the remaining elements three are small and
    one is large, so it decides that the element should be compared with $p$ first.
    Applying strategy $\mathcal{C}$
from right to left uses that of the inspected elements two were small and only one was large, 
so it decides to compare the element with $p$ first, too. Note that the strategies
would differ if, e.\,g., the rightmost element were a medium element.}
    \label{fig:strategy:difference}
    \end{figure}
The same calculation as the proof of the previous theorem shows that  
$\text{ACT}_\mathcal{C} - \text{ACT}_\mathcal{O}$ is $O(\log n)$,
which---as mentioned in the proof of the previous theorem---sums up to a
total additive contribution of $O(n)$ when using strategy $\mathcal{C}$ in a dual-pivot quicksort
algorithm. 
\qed
\end{proof}
Thus, dual-pivot quicksort with strategy $\mathcal{C}$ has
average cost at most $O(n)$ larger than dual-pivot quicksort using the
(unrealistic) optimal strategy $\mathcal{O}$.

We have identified two asymptotically optimal strategies ($\mathcal{SP}$ and $\mathcal{C}$) 
which can be used in an
actual algorithm. Note that both strategies have an additional
summand of $O(n)$. Unfortunately, the solution of the dual-pivot quicksort 
recurrence, cf.~Theorem~\ref{thm:10}, 
does not  give any information about the constant hidden in the $O(n)$ term.
However, the average partitioning cost of strategy $\mathcal{C}$ differs 
by only
$O(\log n)$ from the partitioning cost of the optimal strategy $\mathcal{O}$, while $\mathcal{SP}$ 
differs by $O(n^{1-\varepsilon})$. So one may be led to believe that $\mathcal{C}$ has a (noticeably)
lower average comparison count for real-world input lengths. 
In
Section~\ref{sec:experiments} we will see that differences are clearly
visible in the average comparison count measured in experiments. However, we
will also see that the necessity for  bookkeeping renders implementations of
strategy $\mathcal{C}$ impractical.

\section{Choosing Pivots From a Sample}\label{sec:pivot:sample}
In this section, we consider the variation of quicksort where the pivots
are chosen from a small sample. Intuitively, this guarantees better pivots 
in the sense that  the partition sizes are more balanced.
For classical quicksort, the Median-of-$k$ strategy is optimal
w.r.t. minimizing the average comparison count \cite{MartinezR01}, which means that
the median in a sample of $k$ elements is chosen as the pivot. The standard
implementation of Yaroslavskiy's algorithm in Sun's Java 7 uses an 
intuitive generalization of this strategy: it chooses the two
tertiles in a sample of five elements as pivots.\footnote{
    In an ordered set $S = \{x_1,\ldots,x_k\}$, the two tertiles are the elements of 
rank $\lceil k/3\rceil$ and $\lceil 2k/3\rceil$.}

We will compare dual-pivot quicksort algorithms which use the two tertiles of the 
first five elements of the input as the two pivots with classical quicksort. 
Moreover, we will see that the optimal pivot choices for dual-pivot quicksort 
are not the two tertiles of
a sample, but rather the elements of rank $k/4$ and $k/2$. 

We remark that Wild, Nebel, and Mart\'inez \cite{WildNM14} provide a much more 
detailed study of pivot sampling in Yaroslavskiy's algorithm. 

\subsection{Choosing the Two Tertiles in a Sample of Size $5$ as Pivots}

We sort the first five elements and take the second and fourth elements as pivots.
The probability that $p$ and $q$, $p<q$,
are chosen as pivots is exactly $(s \cdot m \cdot \ell)/\binom{n}{5}.$ Following
Hennequin \cite[pp. 52--53]{hennequin}, for average partitioning cost $\E(P_n)
= a \cdot n + O(1)$ we get 
\begin{equation}\label{eq:1000}
    \E(C_n) = \frac{1}{H_6 - H_2} \cdot a \cdot n \ln n  + O(n)= \frac{20}{19} \cdot a \cdot n \ln n + O(n),
\end{equation} 
where $H_n$ denotes the $n$-th harmonic number. 

When applying Lemma~\ref{lemma:01}, we have average partitioning cost $a
\cdot n + O(n^{1 - \varepsilon})$. Using the same argument as in the proof of
Theorem~\ref{thm:10}, the average comparison count becomes $20/19\cdot a \cdot
n \ln n + O(n)$ in this case. (This a special case of the much more general 
pivot sampling strategy that is described and analyzed in \cite[Theorem 6.2]{WildNM14}.) 

We will now investigate the effect on the average number of key comparisons in
Yaroslavskiy's algorithm and the asymptotically optimal strategy $\mathcal{SP}$ 
from Section~\ref{sec:decreasing}. The average
number of medium elements remains $(n-2)/3$.  For strategy $\mathcal{Y}$, we
calculate (again using Maple\textsuperscript{\textregistered})
\begin{align*}
    \E(P^\mathcal{Y}_n) &= \frac{4}{3}n  + \frac{1}{\binom{n}{5}}
    \sum_{s + \ell \leq n - 5} \frac{\ell \cdot (2s + m) \cdot s \cdot m \cdot
    \ell}{n-5} + O(n^{1 - \varepsilon}) = \frac{34}{21}n + O(n^{1 - \varepsilon}).
\end{align*}
Applying \eqref{eq:1000}, we get $\E(C^\mathcal{Y}_n) = 1.704.. n \ln n + O(n)$ 
key comparisons on average. (Note that \cite{nebel13} calculated
this leading coefficient as well.) This is slightly better than ``clever
quicksort'', which uses the median of a sample of three elements as a single
pivot element and achieves $1.714.. n \ln n + O(n)$  key comparisons on average
\cite{vanEmden}.  For strategy $\mathcal{SP}$, we get (similarly as in Section~\ref{sec:decreasing})
\begin{align*}
    \E(P^\mathcal{SP}_n) &=  \frac{4}{3}n +  \frac{2}{\binom{n}{5}}
    \sum_{\substack{s + \ell \leq n-5\\s \leq \ell}} s \cdot s \cdot m \cdot
    \ell + O(n^{1 - \varepsilon}) = \frac{37}{24}n + O(n^{1 - \varepsilon}).
\end{align*}
%4/3 * n - 5/3 + 1/binomial(n,5) * (sum(sum((p-1) * (p-1) * (q-p) * (n-q),q, p+2, n-p) + sum( (n-q) * (p-1) * (q-p) * (n-q),q,n-p+1, n - 1),p, 2, n/2) + sum(sum((n - q) * (p-1) * (q-p) * (n-q),q, p + 2, n - 1), p, n/2 + 1, n - 1));
Again using \eqref{eq:1000}, we obtain $\E(C^\mathcal{SP}_n) = 1.623.. n \ln
n + O(n)$, improving further
on the leading coefficient compared to clever quicksort and Yaroslavskiy's
algorithm.

\subsection{Pivot Sampling in Classical Quicksort and Dual-Pivot Quicksort}
In the previous subsection, we have shown that optimal dual-pivot quicksort
using a sample of size $5$ clearly beats clever quicksort which uses the median
of three elements. We will now investigate how these two variants compare when
the sample size grows.

The following proposition, which is a special case of \cite[Proposition
III.9 and Proposition III.10]{hennequin}, will help in this
discussion.
\begin{proposition}
    Let $a \cdot n + O(1)$ be the average partitioning cost of a quicksort
    algorithm $\mathcal{A}$ that chooses the pivot(s) from a sample of size $k$,
    for constants $a$ and $k$. Then the following holds:
    \begin{enumerate}
	\item If $k+1$ is even and $\mathcal{A}$ is a classical quicksort
	    variant that chooses the median of these $k$ samples, then the
	    average sorting cost is
	    \begin{align*}
	    \frac{1}{H_{k+1} - H_{(k+1)/2}} \cdot a \cdot n \ln n + O(n).
	    \end{align*}
	\item If $k+1$ is divisible by 3 and $\mathcal{A}$ is a dual-pivot quicksort
	    variant that chooses the two tertiles of these $k$ samples as pivots, then the
	    average sorting cost is
	    \begin{align*}
	    \frac{1}{H_{k+1} - H_{(k+1)/3}} \cdot a \cdot n \ln n + O(n).
	    \end{align*}
    \end{enumerate}
\end{proposition}
Note that for classical quicksort we have partitioning cost $n - 1$.
Thus, the average sorting cost becomes $\frac{1}{H_{k+1} - H_{(k+1)/2}} n \ln n
+ O(n)$.

For dual-pivot partitioning algorithms, the
probability that $p$ and $q$, $p < q$, are chosen as pivots in a sample of size
$k$ where $k+1$ is divisible by $3$ is exactly
\begin{align*}
    \frac{\binom{p-1}{(k-2)/3}\binom{q-p-1}{(k-2)/3}\binom{n-q}{(k-2)/3}}{\binom{n}{k}}.
\end{align*}
Thus, the average partitioning cost $\E(P^\mathcal{SP}_{n,k})$ of strategy $\mathcal{SP}$ using a sample of size
$k$ can be calculated as follows:
\begin{align}\label{eq:700}
    \E(P^\mathcal{SP}_{n,k}) = \frac{4}{3}n +  \frac{2}{\binom{n}{k}}\sum_{s
    \leq \ell}\binom{s}{(k-2)/3}\binom{m}{(k-2)/3}\binom{\ell}{(k-2)/3} \cdot s + O(n^{1 - \varepsilon}).
\end{align}
Unfortunately, we could not find a closed form of $\E(P^\mathcal{SP}_{n,k})$.
Some calculated values in which classical and dual-pivot quicksort with strategy $\mathcal{SP}$
use the same
sample size can be found in Table~\ref{tab:sample:sorting:cost}. These values
clearly indicate that starting from a sample of size $5$ classical quicksort has
a smaller average comparison count than dual-pivot quicksort.\footnote{Note that
this statement does not include that these two variants have different linear terms 
for the same sample size.} This raises the question whether
dual-pivot quicksort is inferior to classical quicksort using the Median-of-$k$ strategy. 

\begin{table}
\tbl{Comparison of the leading term of the average cost of classical quicksort and 
dual-pivot quicksort for specific sample sizes. Note that for real-world input
    sizes, however, the linear term can make a big
    difference.\label{tab:sample:sorting:cost}}
    {
    \begin{tabular}{l||c|c|c|c}
	Sample Size & $5$ & $11$ & $17$ & $41$ \\ \hline
	Median (QS) & $1.622.. n \ln n  $ & $1.531.. n \ln n $ & $1.501.. n \ln n
	    $ & $1.468.. n \ln n $ \\ \hline
        Tertiles (DP QS) & $1.623.. n \ln n $ & $1.545.. n \ln n $ & $1.523.. n \ln n
		   $ & $1.504.. n \ln n $
     \end{tabular}
     }
 \end{table}

\subsection{Optimal Segment Sizes for Dual-Pivot
Quicksort}\label{sec:optimal:segment:sizes}
It is known from, e.g., \cite{MartinezR01} that for classical quicksort in which the pivot is chosen as the median of a
fixed-sized sample, the leading term of the average comparison count converges with increasing sample size 
to the lower bound of
$(1/\ln 2) \cdot n \ln n = 1.4426.. n \ln n$.  
Wild, Nebel, and Mart\'inez observed in \cite{WildNM14} that this is not the case for
Yaroslavskiy's algorithm, which makes at least $1.4931.. n \ln n - O(n)$ comparisons on
average no matter how well the pivots are chosen. In this section, we will show how to match the lower
bound for comparison-based sorting algorithms with a dual-pivot approach.

We study the following setting, which was considered in \cite{MartinezR01,WildNM14}
as well. We assume that for a random input of $n$ elements\footnote{We disregard the two pivots
in the following discussion.}
we can choose (for free) two pivots w.r.t. a vector $\vec{\tau} = (\tau_1, \tau_2, \tau_3)$ 
such that the input contains exactly $\tau_1 n$ small elements, $\tau_2 n$ medium elements, and
$\tau_3 n$ large elements. Furthermore, we consider the (simple) classification
strategy $\mathcal{L}$: \emph{``Always compare with the larger pivot first.''}

The following lemma says that this strategy achieves the minimum possible
average comparison count for comparison-based sorting algorithms, $1.4426.. n
\ln n$, when setting $\tau = (\frac14,\frac14,\frac12)$. 

\begin{lemma}
    Let $\vec{\tau} = (\tau_1, \tau_2, \tau_3)$ with $0 < \tau_i < 1$ and $\sum_{i}\tau_i = 1$, for
    $i \in \{1,2,3\}$. Assume that for each input size $n$ we can choose two pivots such that there are
    exactly $\tau_1 \cdot n$ small, $\tau_2 \cdot n$ medium, and $\tau_3 \cdot n$ large elements. Then
    the comparison count of strategy $\mathcal{L}$ is\footnote{Here, 
    $f(n)\sim g(n)$ means that $\lim_{n \rightarrow \infty} f(n)/g(n) = 1$.}
    \begin{align*}
	p^{\vec{\tau}}(n) \sim \frac{1+\tau_1+\tau_2}{-\sum_{1 \leq i \leq 3}
	\tau_i \ln \tau_i} n \ln n.
    \end{align*}
This value is minimized for $\vec{\tau}^\ast = (1/4, 1/4, 1/2)$ giving 
\begin{align*}
    p^{\vec{\tau}^\ast}\!(n) \sim \left(\frac{1}{\ln 2}\right) n \ln n = 1.4426.. n \ln
    n.
\end{align*}
\end{lemma}

\begin{proof}
    On an input consisting of $n$ elements, strategy $\mathcal{L}$ makes $n +
    (\tau_1 + \tau_2) n$ comparisons. 
    Thus, the comparison count of strategy $\mathcal{L}$ follows the recurrence

    \begin{align*}
	p^{\vec{\tau}}(n) = n + (\tau_1 + \tau_2) n +
	p^{\vec{\tau}}(\tau_1 \cdot n) +
	p^{\vec{\tau}}(\tau_2 \cdot n) +
	p^{\vec{\tau}}(\tau_3 \cdot n).
    \end{align*}
    Using the Discrete Master Theorem from \cite[Theorem 2.3, Case (2.1)]{Roura01}, we obtain the
    following solution for this recurrence:
    \begin{align*}
	p^{\vec{\tau}}(n) \sim \frac{1 + \tau_1 + \tau_2}{-\sum_{i = 1}^{3} \tau_i \ln
	\tau_i} n \ln n.
    \end{align*}
    Using Maple\textsuperscript{\textregistered}, one finds that $p^{\vec\tau}$ is minimized for
    $\vec{\tau}^\ast = (\frac14,\frac14,\frac12)$, giving $p^{\vec{\tau}^\ast}(n)
    \sim
    1.4426.. n \ln n$.\qed 
\end{proof}
The reason why strategy $\mathcal{L}$ with this particular choice of pivots
achieves the lower bound is simple: it makes (almost) the same comparisons as does classical
quicksort using the median of the input as pivot. On an input of
length $n$, strategy $\mathcal{L}$ makes $3/2n$ key comparisons and then
makes three recursive calls to inputs of length $n/4$, $n/4$, $n/2$. On an 
input of length $n$, classical 
quicksort using the median of the input as the pivot makes $n$ comparisons to
split the input into two subarrays of length $n/2$. Now consider only the
recursive call on the left subarray. After $n/2$ comparisons, the input is split
into two subarrays of size $n/4$ each. Now there remain two recursive calls on
two subarrays of size $n/4$, and one recursive call on a subarray of size $n/2$ (the
right subarray of the original input), like in strategy $\mathcal{L}$. Since
classical quicksort using the median of the input clearly makes $n \log n$ key
comparisons, this bound must also hold for strategy $\mathcal{L}$.

\section{Experiments}\label{sec:experiments}
We have implemented the methods presented in this paper in {\Cpp} and Java.
Our experiments were carried out on an Intel i7-2600 at 3.4 GHz 
with 16 GB RAM running Ubuntu 13.10 with kernel version 3.11.0. 
For compiling {\Cpp} code, we used \emph{gcc} in version 4.8.
For compiling Java code, we used Oracle's \emph{Java 8 SDK}. All Java runtime tests
were preceeded by a warmup phase for the just-in-time compiler (JIT), in which we let each algorithm
sort $10\,000$ inputs of length $100\,000$.

For better orientation, the algorithms considered in this section are presented
in Table~\ref{tab:experiments:algorithms}. Pseudocode for
the dual-pivot methods is provided in Appendix~\ref{app:sec:algorithms}. In the following, 
we use a calligraphic letter both for the classification strategy and the actual dual-pivot quicksort algorithm.

\begin{table}[t!]
    \tbl{Overview of the algorithms considered in the experiments.\label{tab:experiments:algorithms}}
    {
    \begin{tabular}{c||c||c||c}
        Abbreviation & Full Name & Strategy& Pseudocode\\ \hline
        $\mathcal{QS}$ & Classical Quicksort & --- & e.g., \cite[Algo. 1]{nebel12} \\ \hline
        $\mathcal{Y}$ & Yaroslavskiy's Algorithm & Section~\ref{sec:methods} & Algorithm~\ref{algo:yaroslavskiy:partition} (Page \pageref{algo:yaroslavskiy:partition}) \\ \hline
        $\mathcal{L}$ & Larger Pivot First & Section~\ref{sec:methods} & Algorithm~\ref{algo:always:q:first} (Page \pageref{algo:always:q:first}) \\ \hline
        $\mathcal{S}$ & Sedgewick's Algorithm (modified) & Section~\ref{sec:methods} & Algorithm~\ref{algo:sedgewick:partition:modified} (Page \pageref{algo:sedgewick:partition:modified}) \\ \hline
        $\mathcal{SP}$ & Sample Algorithm & Section~\ref{sec:decreasing} & Algorithm~\ref{algo:simple:partition} (Page \pageref{algo:simple:partition}) \\ \hline
        $\mathcal{C}$ & Counting Algorithm & Section~\ref{sec:optimal:strategies} & Algorithm~\ref{algo:counting:strategy} (Page
        \pageref{algo:counting:strategy}) \\ \hline
        $\mathcal{K}$ & $3$-Pivot-Algorithm & --- & \cite[Algo. A.1.1]{Kushagra14} 
    \end{tabular}
    }
\end{table}

In Section~\ref{sec:comparison:swap:count}, we experimentally evaluate the
average comparison count of the algorithms considered here. In
Section~\ref{sec:running:times}, we focus on the actual running time needed to
sort a given input.  The charts of our experiments can be found at the end of
this paper.

\subsection{The Average Comparison Count}\label{sec:comparison:swap:count}
We first have a look at the number of comparisons needed to sort
a random input of up to $2^{29}$ integers. We did not switch to a
different sorting algorithm, e.\,g., insertion sort, to sort short subarrays.

Figure~\ref{fig:comp:direct} shows the results of our experiments for algorithms
that choose the pivots directly from the input. We see that for practical values of $n$  
the lower order terms in the average comparison count have a big influence on the
number of comparisons for all algorithms. (E.g., for Yaroslavskiy's algorithm this lower order term 
is dominated by the linear term
$-2.46n$, as known from \cite{nebel12}.)  
Nevertheless, we may conclude that the theoretical studies on the average 
comparison count correspond to what can be observed in practice.
Note that there is a difference between
the optimal strategies $\mathcal{SP}$ and $\mathcal{C}$. 
We also see that the modified version of Sedgewick's algorithm beats
Yaroslavskiy's algorithm, as calculated in Section~\ref{sec:methods}.

Figure~\ref{fig:comp:sample} shows the same experiment for algorithms that
choose the pivots from a small sample. For dual-pivot quicksort variants, we used
two different versions. Algorithms $\mathcal{Y},\mathcal{SP},\mathcal{C},\mathcal{L}$ choose 
the tertiles of a sample of size $5$ as pivots. $\mathcal{QS}$ is classical quicksort using the 
median of three strategy. $\mathcal{Y}'$ is Yaroslavskiy's algorithm with the tertiles of 
a sample of size $11$ as the pivots; $\mathcal{L}'$ is algorithm $\mathcal{L}$ using the 
third- and sixth-largest element in a sample of size $11$ as the pivots. 
This plot confirms the theoretical results from Section~\ref{sec:pivot:sample}. Especially,
the simple strategy $\mathcal{L}$ beats Yaroslavskiy's algorithm for a sample of size $11$ for 
the pivot choices introduced in Section~\ref{sec:optimal:segment:sizes}.

\subsection{Running Times}\label{sec:running:times} We now consider the running
times of the algorithms for sorting a given input.  We restrict our experiments to
sorting random permutations of the integers $\{1,\ldots,n\}$. It remains future
work to compare the algorithms in more detail.

As a basis for comparison with other methods, we also include the three pivot quicksort
algorithm described in \cite{Kushagra14} (about 8\% faster than
Yaroslavskiy's algorithm in their setup), which we call $\mathcal{K}$.
Similarly to classical quicksort, this algorithm uses two pointers to scan the
input from left-to-right and right-to-left until these pointers cross.
Classifications are done in a symmetrical way: First compare to the
middle pivot; then (appropriately) compare either to the largest or
smallest pivot. (In this way, each element is compared to exactly two
pivots.) Elements are moved to a suitable position by the help of two
auxiliary pointers. We note that the pseudocode from \cite[Algorithm
A.1.1]{Kushagra14} uses multiple swaps when fewer assignments are sufficient to
move elements.  We show pseudocode of our implementation in
Appendix~\ref{app:sec:algo:three:pivot}.

In our experiments, subarrays of size at most $16$, $20$, and $23$ were sorted by
insertion sort for classical quicksort, dual-pivot quicksort, and the three
pivot quicksort algorithm, respectively. Furthermore, strategy $\mathcal{SP}$ uses 
strategy $\mathcal{L}$ to sort inputs that contain no more than  $1024$ elements.

\paragraph{{\Cpp}\, Experiments} We first discuss our results on {\Cpp}
code compiled with \emph{gcc}-4.8. Since the used compiler flags might influence the
observed running times, we considered 4 different compiler
settings.  In setting 1 we have compiled the source code with \emph{-O2}, in
setting 2 with \emph{-O2 -funroll-loops}, in setting 3 with \emph{-O3}, and in
setting 4 with \emph{-O3 -funroll-loops}. The option \emph{-funroll-loops}
tells the compiler to unroll (the first few iterations of) a loop. In all
settings, we used \emph{-march=native}, which means that the compiler tries to
optimize for the specific CPU architecture we use during compilation. 

The experiments showed that there is no significant difference between the settings using \emph{-O2} 
and those using \emph{-O3}. So, we just focus on the first two compiler settings which use \emph{-O2}. 
The running time results we obtained in these two settings are shown in
Figure~\ref{fig:running:time:setting:1} and Figure~\ref{fig:running:time:setting:2} at the end of this paper.
We first note that \emph{loop unrolling} makes all algorithms behave slightly worse with respect to running times. 
Since a single innocuous-looking compiler flag may have such an impact on running time, we stress that our results do not
allow final statements about the running time behavior of quicksort variants. In the following, 
we restrict our evaluation to setting $1$. 

With respect to average running time, we get the following results, 
see Figure~\ref{fig:running:time:setting:1} (the discussion uses the measurements obtained for inputs
of size $2^{27}$): 
Yaroslavskiy's algorithm
is the fastest algorithm, but the difference to the dual-pivot algorithm 
$\mathcal{L}$ and the three pivot algorithm $\mathcal{K}$ is negligible. The
asymptotically comparison-optimal sampling algorithm $\mathcal{SP}$ cannot compete with these three 
algorithms w.r.t. running time. On average it is about $7.2\%$ slower than algorithm $\mathcal{Y}$.
Classical quicksort is about $8.3\%$ slower than $\mathcal{Y}$. The slowest algorithm
is the counting algorithm $\mathcal{C}$, on average it is about $14.8\%$ slower than $\mathcal{Y}$.

Now we consider the significance of running time differences. For that, we let each algorithm sort
the same $1\,000$ inputs containing $2^{27}$ items. In Table~\ref{tab:running:times:dual:pivot} we consider 
the number of cases which support the hypothesis that an algorithm is a given percentage faster than another algorithm.
The table shows that the difference in running time is about $1\%$ smaller than the average suggested
if we consider only significant running time differences, i.\,e., differences that were observed in at least 
$95\%$ of the inputs. We conclude that there is  no 
significant difference between the running times of the three fastest quicksort variants.

\begin{table}[th]
    \tbl{
    Comparison of the actual running times of the algorithms on $1000$ different inputs of size $2^{27}$. 
    A table cell in row labelled ``$A$'' and column labelled ``$B$'' contains a string ``$x\%$/$y\%$/$z\%$'' and is read as follows:
    ``In about $95\%$, $50\%$, and $5\%$ of the cases algorithm $A$ was more than $x$, $y$, and $z$ percent faster than algorithm $B$,
    respectively.''\label{tab:running:times:dual:pivot}}
    {
    \begin{tabular}{l||c|c|c|c|c|c}
        & $\mathcal{Y}$ & $\mathcal{L}$ & $\mathcal{K}$ & $\mathcal{SP}$ & $\mathcal{QS}$ & $\mathcal{C}$ \\ \hline 
        $\mathcal{Y}$ & --- & --/--/$0.9\%$ & --/$1.0\%$/$2.1\%$ & $6.4\%$/$7.2\%$/$8.0\%$ & $7.2\%$/$8.3\%$/$9.4\%$ & $13.9\%$/$14.8\%$/$15.8\%$  \\ \hline
        $\mathcal{L}$ & --/--/$0.8\%$ & --- & --/$0.9\%$/$2.1\%$ & $5.8\%$/$7.0\%$/$8.5\%$ & $7.0\%$/$8.1\%$/$9.5\%$ & $13.7\%$/$14.6\%$/$15.8\%$ \\ \hline
        $\mathcal{K}$ &--- & --/--/$0.3\%$ & --- & $4.7\%$/$6.0\%$/$7.6\%$ & $5.9\%$/$7.2\%$/$8.5\%$ & $12.3\%$/$13.6\%$/$14.9\%$ \\ \hline
        $\mathcal{SP} $ & --- & --- & --- & --- & --/$1.0\%$/$2.4\%$ & $5.8\%$/$7.1\%$/$8.4\%$ \\ \hline
        $\mathcal{QS} $ & --- & --- & --- & --- & --- & $4.7\%$/$5.9\%$/$7.2\%$ \\ \hline
    \end{tabular}
}
\end{table}

The good performance of the simple strategy ``always compare to the larger pivot first'' is especially interesting, 
because it is bad from a theoretical point of view: it makes $2 n \ln n$ comparisons and $0.6 n \ln n$ swaps on average. So,
it does not improve in both of these cost measures compared to classical quicksort, but is still faster.
We may try to explain these differences in measured running times by looking at
the average instruction count and the cache behavior of these algorithms. (The
latter is motivated by the cost measure ``cache misses'' considered
in \cite{Kushagra14}.) Table~\ref{table:processor:measurements} shows
measurements of the \emph{average total instruction count} and the \emph{average
number of L1/L2 cache misses},\footnote{Such statistics can be collected, e.g.,
    by using the \emph{performance application programming interface} from {\tt
http://icl.cs.utk.edu/papi/}.} both in total and in relation to algorithm
$\mathcal{K}$.

\begin{table}[t]
    \tbl{Average number of total instructions and cache misses scaled by $n
        \ln n$ for the algorithms considered in this section. In parentheses,
        these figures are set into relation to the values measured for algorithm
        $\mathcal{K}$. (The relative difference has been calculated directly
        from the experimental data.) The figures were
        obtained by sorting $1{\,}000$ inputs of length $2^{27}$ separately by each
algorithm using the compiler flag \emph{-O2}.
\label{table:processor:measurements}}
    {
    \begin{tabular}{c|r|r|r}
        Algorithm & total \#instructions & L1 misses & L2 misses \\ \hline
        $\mathcal{QS}$ & $10.58 n \ln n$ ($+12.7\%$) & $0.142 n \ln n$ ($+47.9\%$) & $0.030 n \ln n$ ($+178.1\%$) \\ 
        $\mathcal{Y}$ & $10.42 n \ln n$ ($+11.0\%$) & $0.111 n \ln n$ ($+15.6\%$) & $0.015 n \ln n$ ($+\phantom{0}39.9\%$) \\ 
        $\mathcal{SP}$ & $10.05 n \ln n$ ($+\phantom{0}7.1\%$) & $0.111 n \ln n$ ($+15.6\%$) & $0.014 n \ln n$ ($+\phantom{0}31.7\%$) \\ 
        $\mathcal{C}$ & $14.08 n \ln n$ ($+50.1\%$) & $0.111 n \ln n$ ($+15.6\%$) & $0.012 n \ln n$ ($+\phantom{0}11.8\%$) \\ 
        $\mathcal{L}$ & $9.50 n \ln n$ ($+\phantom{0}1.2\%$) & $0.111 n \ln n$ ($+15.6\%$) & $0.016 n \ln n$ ($+\phantom{0}49.0\%$) \\ 
        $\mathcal{K}$ & $9.38 n \ln n$ ($\phantom{06}$---$\phantom{0\%}$) & $0.096 n \ln n$ ($\phantom{06}$---$\phantom{0\%}$) & $0.011 n \ln n$ ($\phantom{006}$---$\phantom{0\%}$) \\ 
    \end{tabular}
    }
\end{table}

With respect to the average number of total instructions, we see that strategy
$\mathcal{K}$ needs the fewest instructions on average. It is about $1.2\%$ better
than algorithm $\mathcal{L}$. Furthermore, it is $7.1\%$ better
than the sampling algorithm $\mathcal{SP}$, about $11\%$ better than Yaroslavskiy's
algorithm and about $12.7\%$ better than classical quicksort. 
As expected, the theoretically optimal strategy
$\mathcal{C}$ needs by far the most instructions because of its necessity for
bookkeeping. However, focusing exclusively on
the total instruction count does not predict the observed running time behavior
accurately. (In particular, $\mathcal{Y}$ should be slower than classical quicksort,
which has also been observed in
\cite{nebel13} with respect to Sun's Java 7.)

The authors of \cite{Kushagra14} conjectured that another cost measure---the
\emph{average number of cache misses}---explains observed running time behavior.
From Table~\ref{table:processor:measurements} we see that strategy $\mathcal{K}$
shows the best performance with respect to cache misses, confirming the
theoretical results of \cite{Kushagra14}. In our experiments, we observed that
the relative difference of L1 cache misses between the $1$-, $2$-, and $3$-pivot
algorithms match the theoretical results from \cite{Kushagra14}.  (According to
their calculations, dual-pivot variants should incur about $16\%$ more cache
misses, while classical quicksort should incur about $44\%$ more cache misses.)
This might explain why classical quicksort is the slowest algorithm, but cannot
explain why the dual-pivot algorithms $\mathcal{Y}$ and $\mathcal{L}$ can
compete with the three pivot algorithm.

With respect 
to L$2$ cache misses, we see that the cache misses are much higher than predicted.
Compared to the $3$-pivot algorithm, classical quicksort incurs $178\%$ more
cache misses, and the worst dual-pivot algorithm w.r.t. L$2$ misses has $49\%$ more
cache misses. Since L$2$ misses have a much higher impact on the CPU cycles
spent waiting for memory, this might amplify differences between these algorithms.
In general, we note that only a fraction of about $1\%$ of the instructions lead to cache
misses. Consequently, one has to accurately measure the cycles spent waiting
for memory in these situations. Investigating these issues
is a  possible direction for future work. With respect to other cost measures,
Mart\'inez, Nebel, and Wild \cite{MartinezNW15} analyzed branch misses in classical quicksort and 
Yaroslavskiy's algorithm. Their result is that the difference in the average number of branch mispredictions
between these two algorithms is too small to explain differences in empirical running time.  
Based on our experimental measurements, we believe that 
this cost measure cannot explain running time differences between dual-pivot quicksort algorithms, either. 

We conclude that in the light of our experiments
neither
the average number of instructions
nor the average number of cache misses can fully explain empirical running time.

\paragraph{Java Experiments} Figure~\ref{fig:running:time:java:8} shows the
running time measurements we got using Oracle's \emph{Java 8}. 
Again, there is no significant difference between algorithms $\mathcal{K}, \mathcal{L},$ and $\mathcal{Y}$.
The three pivot quicksort algorithm $\mathcal{K}$ is the fastest
algorithm using Java 8 on our setup. It is about $1.3\%$ faster than Yaroslavskiy's 
algorithm on average. The optimal strategy $\mathcal{C}$ is about  $37\%$ slower
on average than algorithm $\mathcal{K}$. Classical quicksort and the sampling 
strategy $\mathcal{SP}$ are about $7.7\%$ slower on average. 

\section{Conclusion and Open Questions}\label{sec:conclusion}

We have studied dual-pivot quicksort algorithms in a unified way and found
optimal partitioning methods that minimize the average number of key comparisons
up to $O(n)$. This minimum is $1.8 n \ln
n + O(n)$. We showed an optimal pivot choice for a simple dual-pivot
quicksort algorithm and conducted an experimental study of the 
practicability of dual-pivot quicksort algorithms.

Several open questions remain. From a theoretical point of view, one should
generalize our algorithms to the case that three or more pivots are used. At the
moment, the optimal average comparison count for $k$-pivot quicksort is unknown
for $k \geq 3$. This is the subject of ongoing work. From both a theoretical and
an engineering point of view, the most important question is to find a cost
measure that accurately predicts empirical running time behavior of variants of
quicksort algorithms.  Here we could only provide evidence that neither
standard measures like the average comparison count or the average swap count,
nor empirical measures like the average instruction count or the cache behavior
predict running time correctly when considered in isolation.

\medskip

\noindent\textbf{Acknowledgements.} The authors thank Thomas Hotz for a useful
suggestion regarding the optimal strategy of
Section~\ref{sec:optimal:strategies}, and Pascal Klaue for implementing the
algorithms and carrying out initial experiments. The first author thanks Timo
Bingmann, Michael Rink, and Sebastian Wild for interesting and useful
discussions. We thank the referees of the conference and journal submissions for their
insightful comments, which helped a lot in improving the presentation. We especially
thank the referee who pointed out a way of showing that the lower order term 
of most algorithms considered in this paper is $O(n)$,
not only $o(n \ln n)$. 

%\bibliographystyle{splncs03}
%\bibliographystyle{ACM-Reference-Format-Journals}
%\bibliography{lit}

%%% -*-BibTeX-*-
%%% Do NOT edit. File created by BibTeX with style
%%% ACM-Reference-Format-Journals [18-Jan-2012].

\appendix

\section{Missing Details of the Proof of Theorem~\ref{thm:10}}\label{app:proof:thm:10}

Here we solve recurrence \eqref{eq:25} for partitioning cost $\E(P_n)$ at most $K \cdot n^{1-\varepsilon}$.
We use the Continuous Master Theorem of Roura \cite{Roura01}, whose 
statement we will review first.

\begin{theorem}[{\cite[Theorem 18]{Roura01}}]
    Let $F_n$ be recursively defined by 
    \begin{align*}
        F_n = \begin{cases}
            b_n, & \text{for } 0 \leq n < N,\\
           t_n + \sum_{j=0}^{n-1}w_{n,j}F_j, & \text{for } n \geq N,
        \end{cases}
    \end{align*}
    where the toll function $t_n$ satisfies $t_n \sim K n^\alpha \log^\beta(n)$ as $n \rightarrow \infty$
    for constants $K \neq 0, \alpha \geq 0, \beta > -1$. Assume there exists a function $w\colon[0,1] \rightarrow \mathbb{R}$
    such that 
    \begin{align}\label{eq:cmt:shape:function}
        \sum_{j = 0}^{n-1} \left\vert w_{n,j} - \int_{j/n}^{(j+1)/n} w(z) \text{ $dz$} \right\vert = O(n^{-d}),
    \end{align}
    for a constant $d > 0$. Let $H:=1 - \int_0^1 z^\alpha w(z) \text{ $dz$}$. Then we have the 
    following cases:
    \begin{enumerate}
        \item If $H > 0$, then $F_n \sim t_n/H.$
        \item If $H = 0$, then $F_n \sim (t_n \ln n)/\hat{H}$, where
            \begin{align*}
                \hat{H} := - (\beta + 1) \int_0^1 z^\alpha \ln (z) w(z) \text{ $dz$}.
            \end{align*}
        \item If $H < 0$, then $F_n \sim \Theta(n^c)$ for the unique $c \in \mathbb{R}$ with 
            \begin{align*}
                \int_0^1 z^c w(z) \text{ $dz$} = 1.
            \end{align*}
    \end{enumerate}
    \label{thm:CMT}
\end{theorem}
We now solve recurrence \eqref{eq:25} for $t_n = K \cdot n^{1-\varepsilon}$, for some $\varepsilon > 0$. 
    First, observe that recurrence \eqref{eq:25} has weights
    \begin{align*}
        w_{n,j} = \frac{6(n - j - 1)}{n (n-1)}.
    \end{align*}
    We define the shape function $w(z)$ as suggested in \cite{Roura01} by
    \begin{align*}
        w(z) &= \lim_{n \rightarrow \infty} n \cdot w_{n, zn} = 6 (1 - z).
    \end{align*}
    Now we have to check \eqref{eq:cmt:shape:function} to see whether the shape function is suitable. We calculate:
    \begin{align*}
        &\phantom{=} \sum_{j = 0}^{n-1} \left\vert w_{n,j} - \int_{j/n}^{(j+1)/n} w(z) \text{ $dz$} \right\vert\\
        & = 6 \sum_{j = 0}^{n-1} \left\vert\frac{n - j - 1}{n ( n - 1)} - \int_{j/n}^{(j + 1)/n} (1-z) \text{ $dz$} \right\vert\\    
        & = 6 \sum_{j = 0}^{n-1} \left\vert\frac{n - j - 1}{n ( n - 1)} + \frac{2j + 1}{2n^2} - \frac{1}{n}\right\vert\\
        & < 6 \sum_{j = 0}^{n - 1} \left\vert\frac{1}{2n(n-1)} \right\vert = O(1/n).
    \end{align*}
    Thus, $w$ is a suitable shape function.
    By observing that  
    \begin{align*}
        H := 1 - 6 \int_0^1 z^{1-\varepsilon} (1 - z) \text{ $dz$} < 0,
    \end{align*}
    we conclude that the third case of Theorem~\ref{thm:CMT} applies for our recurrence. 
    Consequently, we have to find the unique $c \in \mathbb{R}$ such that 
    \begin{align*}
	6 \int_0^1 z^c (1 - z) \text{ $dz$} = 1,
    \end{align*}
    which is true for $c = 1$. Thus, an average partitioning cost of at most $K \cdot n^{1-\varepsilon}$ 
    yields average sorting cost of $O(n)$.

\section{Dual-Pivot Quicksort: Algorithms in Detail}\label{app:sec:algorithms}

\subsection{General Setup}
The general outline of a dual-pivot quicksort algorithm is presented as Algorithm~\ref{algo:dual:pivot:outline}.

\renewcommand{\alglinenumber}[1]{\footnotesize{#1} }

\begin{algorithm}
    \caption{Dual-Pivot-Quicksort (outline)}\samepage\label{algo:dual:pivot:outline}
    \textbf{procedure} $\textit{Dual-Pivot-Quicksort}$($\textit{A}$, $\textit{left}$, $\textit{right}$)
    \begin{algorithmic}[1]
        \If{$\textit{right} - \textit{left} \leq \textit{THRESHOLD}$}
        \State \textit{InsertionSort}($\textit{A}$, $\textit{left}$, $\textit{right}$);
            \State \Return;
        \EndIf
        \If{$A[\textit{right}] > A[\textit{left}]$}
        \State swap \textit{A}[\textit{left}] and \textit{A}[\textit{right}];
        \EndIf
        \State $\texttt{p} \gets A[\textit{left}]$;
        \State $\texttt{q} \gets A[\textit{right}]$;
        \State \textit{partition}(\textit{A}, \texttt{p}, \texttt{q}, \textit{left}, \textit{right}, $\texttt{pos}_\texttt{p}$, $\texttt{pos}_\texttt{q}$);
        \State \textit{Dual-Pivot-Quicksort}(\textit{A}, \textit{left}, $\texttt{pos}_\texttt{p}$ - 1);
        \State \textit{Dual-Pivot-Quicksort}(\textit{A}, $\texttt{pos}_\texttt{p}$ + 1, $\texttt{pos}_\texttt{q}$ - 1);
        \State \textit{Dual-Pivot-Quicksort}(\textit{A}, $\texttt{pos}_\texttt{q}$ + 1, \textit{right});
\end{algorithmic}
\end{algorithm}
To get an actual algorithm we have to implement a \textit{partition} function
that partitions the input as depicted in Figure~\ref{fig:dual:pivot:partition}.
A partition procedure in this paper has two output variables
$\texttt{pos}_{\texttt{p}}$ and $\texttt{pos}_{\texttt{q}}$ that are used to return
the positions of the two pivots in the partitioned array.

For moving elements around, we make use of the following two operations.

\begin{minipage}[t]{5cm}
\begin{algorithm}[H]
    \textbf{procedure} \textit{rotate3}($\textit{a}, \textit{b}, \textit{c}$)
    \begin{algorithmic}[1]
        \State $\texttt{tmp} \gets \textit{a}$;
        \State $\textit{a} \gets \textit{b}$;
        \State $\textit{b} \gets \textit{c}$;
        \State $\textit{c} \gets \texttt{tmp}$;
    \end{algorithmic}
\end{algorithm}
\end{minipage}
\begin{minipage}[t]{3cm}
    \hfill
\end{minipage}
\begin{minipage}[t]{5cm}
    \begin{algorithm}[H]
    \textbf{procedure} \textit{rotate4}($\textit{a}, \textit{b}, \textit{c}, \textit{d}$)
    \begin{algorithmic}[1]
        \State $\texttt{tmp} \gets \textit{a}$;
        \State $\textit{a} \gets \textit{b}$;
        \State $\textit{b} \gets \textit{c}$;
        \State $\textit{c} \gets \textit{d}$;
        \State $\textit{d} \gets \texttt{tmp}$;
    \end{algorithmic}
    \end{algorithm}
\end{minipage}

\subsection{Yaroslavskiy's Partitioning Method}\label{app:sec:yaroslavskiy}
As mentioned in Section~\ref{sec:methods}, Yaroslavskiy's algorithm makes sure that for
$\ell$ large elements in the input $\ell$ or $\ell - 1$ elements will be compared to
the larger pivot first. How does it accomplish this? By default, it
compares to the smaller pivot first, but for each large
elements that it sees, it will compare the next element to the larger pivot
first.

Algorithm~\ref{algo:yaroslavskiy:partition} shows the partition step of (a slightly modified version of) 
Yaroslavskiy's algorithm. In contrast to the algorithm
studied in \cite{nebel12}, it saves an index check at Line~8, and uses a \texttt{rotate3} operation
to save assignments. (In our experiments this makes Yaroslavskiy's algorithm about $4\%$ faster.)

\begin{algorithm}
    \caption{Yaroslavskiy's Partitioning Method}\samepage\label{algo:yaroslavskiy:partition}
		\textbf{procedure} \textit{Y-Partition}($\textit{A}$, $\textit{p}$, $\textit{q}$, $\textit{left}$,
		$\textit{right}$, $\textit{pos}_{\textit{p}}$, $\textit{pos}_{\textit{q}}$)
    \begin{algorithmic}[1]
		\State $\texttt{l} \gets \textit{left} +1; \texttt{g}
		\gets \textit{right} - 1; \texttt{k} \gets \texttt{l}$;
                \While{$\texttt{k} \leq \texttt{g}$}
                    \If{$\textit{A}[\texttt{k}] < p$}
                        \State swap $\textit{A}[\texttt{k}]$ and $\textit{A}[\texttt{l}]$;
                        \State $\texttt{l} \gets\texttt{l} + 1$;
                    \Else
                        \If{$\textit{A}[\texttt{k}] > q$}
                            \While{$\textit{A}[\texttt{g}] > q$} 
                                \State $\texttt{g} \gets \texttt{g} - 1$;
                            \EndWhile
                            \If{$\texttt{k} < \texttt{g}$}
                                \If{$\textit{A}[\texttt{g}] < \textit{p}$}
                                    \State $\textit{rotate3}(\textit{A}[\texttt{g}], 
                                    \textit{A}[\texttt{k}], \textit{A}[\texttt{l}]);$
                                    \State $\texttt{l} \gets \texttt{l} + 1;$
                                \Else
                                    \State swap $\textit{A}[\texttt{k}]$ and $\textit{A}[\texttt{g}]$;
                                \EndIf
                                \State $\texttt{g} \gets \texttt{g} - 1;$
                            \EndIf
                        \EndIf
                    \EndIf
                    \State $\texttt{k} \gets \texttt{k} + 1$;
		\EndWhile
                \State swap $\textit{A}[\textit{left}]$ and $\textit{A}[\texttt{l}-1]$;
		\State swap  $\textit{A}[\textit{right}]$ and $\textit{A}[\texttt{g}+1]$;
	        \State $\textit{pos}_{\textit{p}} \gets \texttt{l} - 1;  \textit{pos}_{\textit{q}} \gets \texttt{g} + 1;$
    \end{algorithmic}
\end{algorithm}

\subsection{Algorithm Using ``Always Compare to the Larger Pivot First''}
\label{app:sec:algo:larger:first}
Algorithm~\ref{algo:always:q:first} presents an implementation of strategy $\mathcal{L}$ (\emph{``Always compare
to the larger pivot first''}). Like Yaroslavskiy's algorithm, it uses three pointers into the array.
One pointer is used to scan the array from left to right until a large element has been
found (moving small elements to a correct position using the second pointer on the way). Subsequently,
it scans the array from right to left using the third pointer until 
a non-large element has been found. These two elements
are then placed into a correct position. This is repeated until the two pointers have crossed. The design goal
is to check as rarely as possible if these two pointers have met, since this event occurs infrequently. (In contrast,
Yaroslavskiy's algorithm checks this for each element scanned by index $\mathtt{k}$ in Algorithm~\ref{algo:yaroslavskiy:partition}.)

This strategy makes $2n \ln n$ comparisons and $1.6 n \ln n$ assignments on average.

\begin{algorithm}
    \caption{Always Compare To Larger Pivot First Partitioning}\samepage\label{algo:always:q:first}
    \textbf{procedure} \textit{L-Partition}($\textit{A}$,
		$\textit{p}$, $\textit{q}$, $\textit{left}$,
		$\textit{right}$, $\textit{pos}_{\textit{p}}$, $\textit{pos}_{\textit{q}}$)
    \begin{algorithmic}[1]
        \State $\texttt{i} \gets \textit{left} + 1; \texttt{k} \gets \textit{right} - 1; \texttt{j} \gets \texttt{i};$
        \While{$\texttt{j} \leq \texttt{k}$}
            \While{$\textit{q} < \textit{A}[\texttt{k}]$}
                \State $\texttt{k} \gets \texttt{k} - 1$;
            \EndWhile
            \While{$\textit{A}[\texttt{j}] < \textit{q}$}
                \If{$\textit{A}[\texttt{j}] < \textit{p}$}
                    \State swap $\textit{A}[\texttt{i}]$ and $\textit{A}[\texttt{j}]$;
                    \State $\texttt{i} \gets \texttt{i} + 1$;
                \EndIf
                \State $\texttt{j}\gets \texttt{j} + 1$;
            \EndWhile
            \If{$\texttt{j} < \texttt{k}$}
                \If{$\textit{A}[\texttt{k}] > \textit{p}$}
                \State \textit{rotate3}($\textit{A}[\texttt{k}]$, $\textit{A}[\texttt{j}]$, 
                $\textit{A}[\texttt{i}]$);
                    \State $\texttt{i}\gets \texttt{i} + 1$;
                \Else
                    \State swap $\textit{A}[\texttt{j}]$ and $\textit{A}[\texttt{k}]$;
                \EndIf
                \State $\texttt{k}\gets \texttt{k} - 1$;
            \EndIf
        \State $\texttt{j}\gets \texttt{j} + 1;$
        \EndWhile
        \State swap $\textit{A}[\textit{left}]$ and $\textit{A}[\texttt{i}-1]$;
        \State swap $\textit{A}[\textit{right}]$ and $\textit{A}[\texttt{k}+1]$;
        \State $\textit{pos}_{\textit{p}} \gets \texttt{i} - 1;  \textit{pos}_{\textit{q}} \gets \texttt{k} + 1$;
    \end{algorithmic}
\end{algorithm}
\subsection{Partitioning Methods Based on Sedgewick's Algorithm}\label{app:sec:sedgewick}
Algorithm~\ref{algo:sedgewick:partition} shows Sedgewick's partitioning method
as studied in \cite{sedgewick}.

Sedgewick's partitioning method uses two
pointers $\texttt{i}$ and $\texttt{j}$ to scan through the input. 
It does not swap
entries in the strict sense, but rather has two ``holes'' at positions
$\texttt{i}_1$ resp. $\texttt{j}_1$ that can be filled with small resp. large
elements. ``Moving a hole'' is not a swap operation in the strict sense (three elements are involved), but requires the
same amount of work as a swap operation (in which we have to save the content of a
variable into a temporary variable \cite{sedgewick}). An intermediate step in
the partitioning algorithm is depicted in Figure~\ref{fig:sedgewick:layout}.

\begin{figure}[tb]
    \centering
    \scalebox{0.6}{\includegraphics{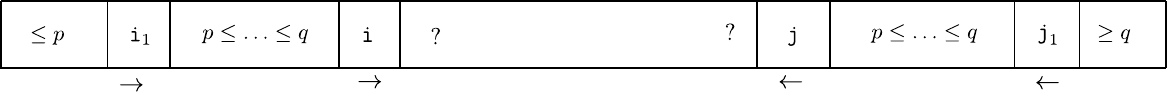}}
    \caption{An intermediate partitioning step in Sedgewick's algorithm.}
    \label{fig:sedgewick:layout}
\end{figure}

The algorithm works as follows: Using $\texttt{i}$
it scans the input
from left to right until it has found a large element, always
comparing to the larger pivot first. Small elements found in this way are moved
to a correct final position using the hole at array position $\texttt{i}_1$. Subsequently, using $\texttt{j}$ it scans
the input from right to left until it has found a small element, always
comparing to the smaller pivot first. Large elements found in this way are moved
to a correct final position using the hole at array position $\texttt{j}_1$. Now it exchanges
the two elements at positions $\texttt{i}$ resp. $\texttt{j}$ and
continues until $\texttt{i}$ and $\texttt{j}$ have met.

\begin{algorithm}
    \caption{Sedgewick's Partitioning Method}\samepage\label{algo:sedgewick:partition}
		\textbf{procedure} \textit{S-Partition}($\textit{A}$,
		$\textit{p}$, $\textit{q}$, $\textit{left}$,
		$\textit{right}$, $\textit{pos}_{\textit{p}}$, $\textit{pos}_{\textit{q}}$)
                \begin{algorithmic}[1]
		  \State $\texttt{i} \gets \texttt{i}_1 \gets \textit{left}; \texttt{j} \gets \texttt{j}_1 : = \textit{right};$
                  \While{\textbf{true}}
                  \State $\texttt{i} \gets \texttt{i} + 1;$
                  \While{$\textit{A}[\texttt{i}] \leq \textit{q}$}
		  \If{$\texttt{i} \geq \texttt{j}$} 
                  \State \textbf{break} outer while 
                  \EndIf
		  \If{$\textit{A}[\texttt{i}] <  \textit{p}$} 
                  \State $\textit{A}[\texttt{i}_1] \gets \textit{A}[\texttt{i}]; \texttt{i}_1 \gets \texttt{i}_1 + 1; \textit{A}[\texttt{i}] \gets \textit{A}[\texttt{i}_1]$;
                  \EndIf
		  \State $\texttt{i} \gets \texttt{i} + 1$;
		 \EndWhile
		 \State $\texttt{j} \gets \texttt{j} - 1;$
                 \While{$\textit{A}[\texttt{j}] \geq \textit{p}$}
                 \If{$\textit{A}[\texttt{j}] >  \textit{q}$} 
                 \State $\textit{A}[\texttt{j}_1] \gets \textit{A}[\texttt{j}]; \texttt{j}_1 \gets \texttt{j}_1 - 1; \textit{A}[\texttt{j}] \gets \textit{A}[\texttt{j}_1]$; 
                 \EndIf
                 \If{$\texttt{i} \geq \texttt{j}$} 
                 \State \textbf{break} outer while
                 \EndIf
		 \State $\texttt{j} \gets \texttt{j} - 1$;
		 \EndWhile
		 \State $\textit{A}[\texttt{i}_1] \gets \textit{A}[\texttt{j}]; \textit{A}[\texttt{j}_1] \gets \textit{A}[i];$
		 \State $\texttt{i}_1 \gets \texttt{i}_1 + 1; \texttt{j}_1 \gets \texttt{j}_1 - 1;$
		 \State $\textit{A}[\texttt{i}] \gets \textit{A}[\texttt{i}_1]; \textit{A}[\texttt{j}] \gets \textit{A}[\texttt{j}_1];$
		 \EndWhile
		 \State $\textit{A}[\texttt{i}_1] \gets \textit{p}; \textit{A}[\texttt{j}_1] \gets \textit{q};$
		 \State $\textit{pos}_{\textit{p}} \gets \texttt{i}_1;  \textit{pos}_{\textit{q}} \gets \texttt{j}_1;$
            \end{algorithmic}
        \end{algorithm}
Algorithm~\ref{algo:sedgewick:partition:modified} shows an implementation of the
modified partitioning strategy from Section~\ref{sec:sedgewick}. In the same way
as Algorithm~\ref{algo:sedgewick:partition} it scans the input from left to
right
until it has found a large element. However, it uses the smaller pivot for the
first comparison in this part. Subsequently, it scans the input from right
to left until it has found a small element. Here, it uses the larger pivot for
the first comparison. 

\begin{algorithm}
    \caption{Sedgewick's Partitioning Method, modified}\samepage\label{algo:sedgewick:partition:modified}
    \textbf{procedure} \textit{S2-Partition}($\textit{A}$,
		$\textit{p}$, $\textit{q}$, $\textit{left}$,
		$\textit{right}$, $\textit{pos}_{\textit{p}}$, $\textit{pos}_{\textit{q}}$)
    \begin{algorithmic}[1]
		  \State $\texttt{i} \gets \texttt{i}_1 \gets \textit{left}; \texttt{j} \gets \texttt{j}_1 : = \textit{right};$
		  \While{\textbf{true}}
                  \State $\texttt{i} \gets \texttt{i} + 1;$
                  \While{\textbf{true}}
                  \If{$\texttt{i} \geq \texttt{j}$}
                  \State \textbf{break} outer while
                  \EndIf
                  \If{$\textit{A}[\texttt{i}] <  \textit{p}$} 
                \State $\textit{A}[\texttt{i}_1] \gets \textit{A}[\texttt{i}]; \texttt{i}_1 \gets \texttt{i}_1 + 1; \textit{A}[\texttt{i}] \gets \textit{A}[\texttt{i}_1]$;
                 \ElsIf{$\textit{A}[\texttt{i}] >  \textit{q}$}
                 \State \textbf{break} inner while
                 \EndIf
		 \State $\texttt{i} \gets \texttt{i} + 1$;
		 \EndWhile
                 \State $\texttt{j} \gets \texttt{j} - 1;$
                 \While{\textbf{true}}
                 \If{$\textit{A}[\texttt{j}] >  \textit{q}$}
                 \State $\textit{A}[\texttt{j}_1] \gets \textit{A}[\texttt{j}]; \texttt{j}_1 \gets \texttt{j}_1 - 1; \textit{A}[\texttt{j}] \gets \textit{A}[\texttt{j}_1]$;
                 \ElsIf{$\textit{A}[\texttt{j}] <\textit{p}$} 
                 \State \textbf{break} inner while
                 \EndIf
                 \If{$\texttt{i} \geq \texttt{j}$}
                 \State \textbf{break} outer while
                 \EndIf
		 \State $\texttt{j} \gets \texttt{j} - 1$;
		 \EndWhile
                 \State $\textit{A}[\texttt{i}_1] \gets \textit{A}[\texttt{j}]; \textit{A}[\texttt{j}_1] \gets \textit{A}[i];$
		 \State $\texttt{i}_1 \gets \texttt{i}_1 + 1; \texttt{j}_1 \gets \texttt{j}_1 - 1;$
		 \State $\textit{A}[\texttt{i}] \gets \textit{A}[\texttt{i}_1]; \textit{A}[\texttt{j}] \gets \textit{A}[\texttt{j}_1];$
		 \EndWhile
                 \State $\textit{A}[\texttt{i}_1] \gets \textit{p}; \textit{A}[\texttt{j}_1] \gets \textit{q};$
		 \State $\textit{pos}_{\textit{p}} \gets \texttt{i}_1;  \textit{pos}_{\textit{q}} \gets \texttt{j}_1;$
            \end{algorithmic}
            \end{algorithm}

\subsection{Algorithms for the Sampling Partitioning Method}\label{app:sec:our:algorithms}

The sampling method $\mathcal{SP}$ from Section~\ref{sec:decreasing} uses a mix of two
classification algorithms: \emph{``Always compare to the smaller pivot first''},
and \emph{``Always compare to the larger pivot first''}.  
The actual partitioning method uses
Algorithm~\ref{algo:always:q:first} for the first $\Samplesize = n / 1024$
classifications and then decides which pivot should be used for
the first comparison in the remaining input. (This is done by comparing the two variables \texttt{i} and
\texttt{k} in Algorithm~\ref{algo:always:q:first}.) If there are more large
elements than small elements in the sample it continues using
Algorithm~\ref{algo:always:q:first}, otherwise it uses
Algorithm~\ref{algo:simple:partition} below. If the input contains fewer than $1024$ elements, 
Algorithm~\ref{algo:always:q:first} is used directly.

\begin{algorithm}
    \caption{Simple Partitioning Method (smaller pivot first)}\samepage\label{algo:simple:partition}
		\textbf{procedure} \textit{SimplePartitionSmall}($\textit{A}$,
		$\textit{p}$, $\textit{q}$, $\textit{left}$,
		$\textit{right}$, $\textit{pos}_{\textit{p}}$, $\textit{pos}_{\textit{q}}$)
                \begin{algorithmic}[1]
		\State $\texttt{l} \gets \textit{left} +1; \texttt{g} \gets \textit{right} - 1; \texttt{k} \gets \texttt{l}$;
                \While{$\texttt{k} \leq \texttt{g}$}
                \If{$\textit{A}[\texttt{k}] < p$}
		\State swap $\textit{A}[\texttt{k}]$ and $\textit{A}[\texttt{l}]$;
		\State $\texttt{l} \gets \texttt{l} + 1$;
		\State $\texttt{k} \gets \texttt{k} + 1$;
		\Else
                \If{$\textit{A}[\texttt{k}] < q$}
		\State $\texttt{k} \gets \texttt{k} + 1$;
		\Else
		\State swap $\textit{A}[\texttt{k}]$ and $\textit{A}[\texttt{g}]$
		\State $\texttt{g} \gets \texttt{g} - 1$;
		\EndIf
                \EndIf
                \EndWhile
		\State swap $\textit{A}[\textit{left}]$ and $\textit{A}[\texttt{l}-1]$;
		\State swap $\textit{A}[\textit{right}]$ and $\textit{A}[\texttt{g}+1]$;
		\State $\textit{pos}_{\textit{p}} \gets \texttt{l} - 1;  \textit{pos}_{\textit{q}} \gets \texttt{g} + 1$;
            \end{algorithmic}
    \end{algorithm}

            \subsection{Algorithm for the Counting Strategy}
            \label{app:sec:algo:counting:strategy}
            Algorithm~\ref{algo:counting:strategy} is an implementation
            of the counting strategy from Section~\ref{sec:optimal:strategies}. 
            It uses a variable $\texttt{d}$ which stores the difference of
            the number of small elements and the number of large elements
            which have been classified so far.  On average this algorithm makes
            $1.8 n \ln n + O(n)$ comparisons and $1.6 n \ln n$ assignments.
\begin{algorithm}
    \caption{Counting Strategy $\mathcal{C}$}\samepage\label{algo:counting:strategy}
    \textbf{procedure} \textit{C-Partition}($\textit{A}$,
		$\textit{p}$, $\textit{q}$, $\textit{left}$,
		$\textit{right}$, $\textit{pos}_{\textit{p}}$, $\textit{pos}_{\textit{q}}$)
    \begin{algorithmic}[1]
        \State $\texttt{i} \gets \textit{left} + 1; \texttt{k} \gets \textit{right} - 1; \texttt{j} \gets \texttt{i};$
        \State $\texttt{d} \gets 0;$ \Comment{$\texttt{d}$ holds the difference of the number of small
        and large elements.}
        \While{$\texttt{j} \leq \texttt{k}$}
            \If{$\texttt{d} > 0$}
                \If{$\textit{A}[\texttt{j}] < \textit{p}$}
                    \State swap $\textit{A}[\texttt{i}]$ and $\textit{A}[\texttt{j}]$;
                    \State $\texttt{i} \gets \texttt{i} + 1;
                            \texttt{j} \gets \texttt{j} + 1;
                            \texttt{d} \gets \texttt{d} + 1;$
                \Else
                    \If{$\textit{A}[\texttt{j}] < \textit{q}$}
                        \State $\texttt{j} \gets \texttt{j} + 1;$
                    \Else
                        \State swap $\textit{A}[\texttt{j}]$ and $\textit{A}[\texttt{k}]$;
                        \State $\texttt{k} \gets \texttt{k} - 1;
                                \texttt{d} \gets \texttt{d} - 1;$
                    \EndIf
                \EndIf
            \Else
                \While{$\textit{A}[\texttt{k}] > \textit{q}$}
                    \State $\texttt{k} \gets \texttt{k} - 1;
                            \texttt{d} \gets \texttt{d} - 1;$
                \EndWhile
                \If{$\texttt{j} \leq \texttt{k}$}
                    \If{$\textit{A}[\texttt{k}] < \textit{p}$}
                        \State \emph{rotate3}($\textit{A}[\texttt{k}], \textit{A}[\texttt{j}],
                                        \textit{A}[\texttt{i}]$);
                        \State $\texttt{i} \gets \texttt{i} + 1$;
                         $\texttt{d} \gets \texttt{d} + 1$;
                    \Else
                        \State swap $\textit{A}[\texttt{j}]$ and $\textit{A}[\texttt{k}]$;
                    \EndIf
                    \State $\texttt{j} \gets \texttt{j} + 1;$
                \EndIf
            \EndIf
        \EndWhile
        \State swap $\textit{A}[\textit{left}]$ and $\textit{A}[\texttt{i}-1]$;
        \State swap $\textit{A}[\textit{right}]$ and $\textit{A}[\texttt{k}+1]$;
        \State $\textit{pos}_{\textit{p}} \gets \texttt{i} - 1;  \textit{pos}_{\textit{q}} \gets \texttt{k} + 1$;
    \end{algorithmic}
\end{algorithm}

\section{A Fast Three-Pivot Algorithm}\label{app:sec:algo:three:pivot}
We give here the complete pseudocode for the three-pivot algorithm described in \cite{Kushagra14}. 
In contrast to the pseudocode given in \cite[Algorithm A.1.1]{Kushagra14}, we removed two unnecessary
bound checks (Line~$5$ and Line~$10$ in our code) and we move misplaced elements
in Lines~$15$--$29$ using less assignments. This is used in the implementation
of \cite{Kushagra14}, as well.\footnote{Code made available by Alejandro
L{\'o}pez-Ortiz.} On average, this algorithm makes $1.846..n \ln n + O(n)$
comparisons 
and $1.57.. n \ln n + O(n)$ assignments. 

\begin{algorithm}
    \caption{Symmetric Three-Pivot Sorting Algorithm}\samepage\label{algo:three:pivot}
    \textbf{procedure} \textit{3-Pivot}($\textit{A}$, $\textit{left}$, $\textit{right}$)
    \begin{algorithmic}[1]
        \Require $\textit{right} - \textit{left} \geq 2$, 
            $\textit{A}[\textit{left}] \leq \textit{A}[\textit{left} + 1] \leq \textit{A}[\textit{right}]$
        \State $\texttt{p}_1 \gets \textit{A}[\textit{left}]; 
        \texttt{p}_2 \gets \textit{A}[\textit{left} + 1];
        \texttt{p}_3 \gets \textit{A}[\textit{right}];$
        \State $\texttt{i} \gets \textit{left} + 2; \texttt{j} \gets \texttt{i};$
        $\texttt{k} \gets \textit{right} - 1; \texttt{l} \gets \texttt{k};$
        \While{$\texttt{j} \leq \texttt{k}$}
            \While{$\textit{A}[j] < \texttt{p}_2$}
                \If{$\textit{A}[j] < \texttt{p}_1$}
                    \State swap $\textit{A}[\texttt{i}]$ and $\textit{A}[\texttt{j}]$;
                    \State $\texttt{i} \gets \texttt{i} + 1;$
                \EndIf
                \State $\texttt{j} \gets \texttt{j} + 1;$
            \EndWhile
            \While{$\textit{A}[k] > \texttt{p}_2$}
                \If{$\textit{A}[k] > \texttt{p}_3$}
                    \State swap $\textit{A}[\texttt{k}]$ and $\textit{A}[\texttt{l}]$;
                    \State $\texttt{l} \gets \texttt{l} - 1;$
                \EndIf
                \State $\texttt{k} \gets \texttt{k} - 1;$
            \EndWhile
            \If{$ \texttt{j} \leq \texttt{k}$}
                \If{$\textit{A}[j] > \texttt{p}_3$}
                    \If{$\textit{A}[k] < \texttt{p}_1$}
                        \State \textit{rotate4}($\textit{A}[\texttt{j}], \textit{A}[\texttt{i}], 
                        \textit{A}[\texttt{k}], \textit{A}[\texttt{l}]$);
                        \State $\texttt{i} \gets \texttt{i} + 1$;
                    \Else
                        \State \textit{rotate3}($\textit{A}[\texttt{j}], 
                         \textit{A}[\texttt{k}], \textit{A}[\texttt{l}]$);
                    \EndIf
                    \State $\texttt{l} \gets \texttt{l} - 1;$
                \Else
                    \If{$\textit{A}[k] < \texttt{p}_1$}
                        \State \textit{rotate3}($\textit{A}[\texttt{j}], \textit{A}[\texttt{i}], 
                        \textit{A}[\texttt{k}]$);
                        \State $\texttt{i} \gets \texttt{i} + 1$;
                    \Else
                        \State swap $\textit{A}[\texttt{j}]$ and 
                        $\textit{A}[\texttt{k}]$;
                    \EndIf
                \EndIf
                \State $\texttt{j} \gets \texttt{j} + 1; \texttt{k} \gets \texttt{k} - 1$;
            \EndIf
        \EndWhile
        \State \textit{rotate3}($\textit{A}[\textit{left} + 1]$, $\textit{A}[\texttt{i} - 1]$,
                                $\textit{A}[\texttt{j} - 1]$);
        \State swap $\textit{A}[\textit{left}]$ and $\textit{A}[\texttt{i} - 2]$;
        \State swap $\textit{A}[\textit{right}]$ and $\textit{A}[\texttt{l} - 1]$;
        \State \textit{3-Pivot}(\textit{A}, \textit{left}, $\texttt{i} - 3$);
        \State \textit{3-Pivot}(\textit{A}, $\texttt{i} - 1$, $\texttt{j} - 2$);
        \State \textit{3-Pivot}(\textit{A}, $\texttt{j}$, $\texttt{l}$);
        \State \textit{3-Pivot}(\textit{A}, $\texttt{l} + 2$, $\textit{right}$);
    \end{algorithmic}
\end{algorithm}

\newpage

\begin{figure}
    \centering
\begin{tikzpicture}
  \begin{axis}[
    xlabel={Items [$\log_2(n)$]},
    ylabel={Comparisons $/ n \ln n$},
    height=8cm,
    width=13cm,
    legend style = { at = {(0.5,0.15)}, anchor=west, draw=none},
    cycle list name = black white,
    legend columns = 2
    ]
    % -IMPORT-DATA statskt2comp res/ktinfty-comparison-count.txt
    % -MULTIPLOT(algo) SELECT LOG(2,items) AS x, AVG( comparisons )/ (items * LN(items)) AS y, MULTIPLOT FROM statskt2comp GROUP BY MULTIPLOT,x ORDER BY MULTIPLOT,x
    \addplot coordinates { (9.0,1.78865) (10.0,1.79511) (11.0,1.80848) (12.0,1.8165) (13.0,1.84026) (14.0,1.84536) (15.0,1.85295) (16.0,1.86175) (17.0,1.87164) (18.0,1.88463) (19.0,1.88326) (20.0,1.89463) (21.0,1.89212) (22.0,1.8982) (23.0,1.90576) (24.0,1.90646) (25.0,1.91415) (26.0,1.91586) (27.0,1.92124) (28.0,1.92075) (29.0,1.92972) };
    \addlegendentry{$\mathcal{QS}$};
    \addplot coordinates { (9.0,1.45758) (10.0,1.49378) (11.0,1.52517) (12.0,1.55546) (13.0,1.57976) (14.0,1.60364) (15.0,1.62469) (16.0,1.64334) (17.0,1.65741) (18.0,1.67506) (19.0,1.68142) (20.0,1.69266) (21.0,1.70593) (22.0,1.7126) (23.0,1.72189) (24.0,1.72627) (25.0,1.734) (26.0,1.74426) (27.0,1.74871) (28.0,1.75167) (29.0,1.75624) };
    \addlegendentry{$\mathcal{Y}$};
    \addplot coordinates { (9.0,1.41905) (10.0,1.45543) (11.0,1.48845) (12.0,1.52313) (13.0,1.5485) (14.0,1.57257) (15.0,1.59275) (16.0,1.6127) (17.0,1.62554) (18.0,1.63843) (19.0,1.64904) (20.0,1.66139) (21.0,1.67216) (22.0,1.67947) (23.0,1.68814) (24.0,1.69448) (25.0,1.70266) (26.0,1.70795) (27.0,1.7169) (28.0,1.71934) (29.0,1.72985) };
    \addlegendentry{$\mathcal{S}$};
    \addplot coordinates { (9.0,1.36638) (10.0,1.40423) (11.0,1.43273) (12.0,1.46405) (13.0,1.48812) (14.0,1.51163) (15.0,1.53226) (16.0,1.54852) (17.0,1.56468) (18.0,1.58044) (19.0,1.58583) (20.0,1.59903) (21.0,1.61074) (22.0,1.61627) (23.0,1.62672) (24.0,1.63171) (25.0,1.63862) (26.0,1.64608) (27.0,1.6527) (28.0,1.65553) (29.0,1.66188) };
    \addlegendentry{$\mathcal{C}$};
    \addplot coordinates { (9.0,1.44368) (10.0,1.47899) (11.0,1.51176) (12.0,1.5384) (13.0,1.55507) (14.0,1.57245) (15.0,1.58835) (16.0,1.59667) (17.0,1.61334) (18.0,1.6218) (19.0,1.63937) (20.0,1.64333) (21.0,1.65426) (22.0,1.65839) (23.0,1.66621) (24.0,1.66821) (25.0,1.67469) (26.0,1.68026) (27.0,1.68387) (28.0,1.68803) (29.0,1.69487) };
    \addlegendentry{$\mathcal{SP}$};
    \addplot coordinates { (9.0,1.57976) (10.0,1.61343) (11.0,1.64196) (12.0,1.66882) (13.0,1.69075) (14.0,1.71679) (15.0,1.73799) (16.0,1.75374) (17.0,1.76826) (18.0,1.78555) (19.0,1.79182) (20.0,1.80236) (21.0,1.81389) (22.0,1.82195) (23.0,1.83203) (24.0,1.83562) (25.0,1.83976) (26.0,1.85202) (27.0,1.85598) (28.0,1.85618) (29.0,1.8602) };
    \addlegendentry{$\mathcal{L}$};
\end{axis}
\end{tikzpicture}
    \caption{Average comparison count (scaled by $n \ln n$) needed to sort
    a random input of up to $n = 2^{29}$ integers. We compare classical
    quicksort ($\mathcal{QS}$), Yaroslavskiy's algorithm ($\mathcal{Y}$), the
    sampling algorithm ($\mathcal{SP}$), the counting algorithm ($\mathcal{C}$),
    the modified version of Sedgewick's algorithm ($\mathcal{S}$), and algorithm
    $\mathcal{L}$. Each data point is the average over $400$ trials.}
    \label{fig:comp:direct}
\end{figure}
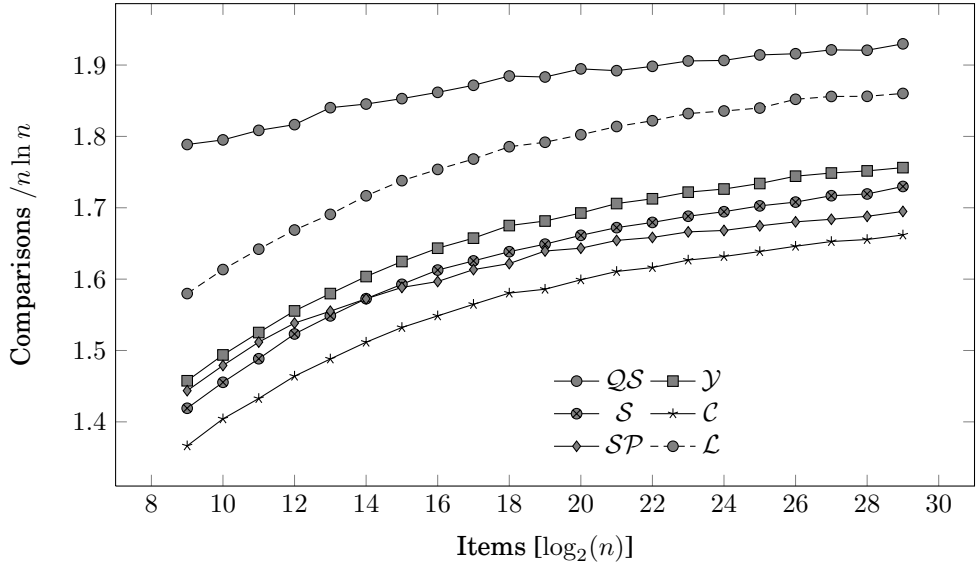

\begin{figure}
    \centering
\begin{tikzpicture}
  \begin{axis}[
    xlabel={Items [$\log_2(n)$]},
    ylabel={Comparisons $/ n \ln n$},
    height=8cm,
    width=12cm,
    legend style = { at = {(0.6,0.25)}, anchor=west, draw=none},
    cycle list name = black white,
    legend columns = 2
    ]
    % -IMPORT-DATA statskt2compsample res/ktinfty-comparison-count-sample.txt
    % -MULTIPLOT(algo) SELECT LOG(2,items) AS x, AVG( comparisons )/ (items * LN(items)) AS y, MULTIPLOT FROM statskt2compsample GROUP BY MULTIPLOT,x ORDER BY MULTIPLOT,x
    \addplot coordinates { (9.0,1.42665) (10.0,1.45338) (11.0,1.47645) (12.0,1.49746) (13.0,1.51036) (14.0,1.5264) (15.0,1.53518) (16.0,1.54957) (17.0,1.55918) (18.0,1.56989) (19.0,1.5747) (20.0,1.58247) (21.0,1.58809) (22.0,1.59361) (23.0,1.60111) (24.0,1.60402) (25.0,1.60878) (26.0,1.61369) (27.0,1.61573) (28.0,1.61976) (29.0,1.62307) };
    \addlegendentry{$\mathcal{QS}$};
    \addplot coordinates { (9.0,1.15571) (10.0,1.20605) (11.0,1.25125) (12.0,1.2891) (13.0,1.31922) (14.0,1.3482) (15.0,1.3711) (16.0,1.39228) (17.0,1.41109) (18.0,1.42744) (19.0,1.44005) (20.0,1.45435) (21.0,1.46724) (22.0,1.47737) (23.0,1.48797) (24.0,1.49587) (25.0,1.50397) (26.0,1.51237) (27.0,1.52037) (28.0,1.52682) (29.0,1.53194) };
    \addlegendentry{$\mathcal{Y}$};
    \addplot coordinates { (9.0,1.19402) (10.0,1.24248) (11.0,1.28258) (12.0,1.31495) (13.0,1.3431) (14.0,1.36743) (15.0,1.38773) (16.0,1.40774) (17.0,1.42345) (18.0,1.43888) (19.0,1.45104) (20.0,1.4635) (21.0,1.47314) (22.0,1.48305) (23.0,1.49217) (24.0,1.50007) (25.0,1.50677) (26.0,1.51435) (27.0,1.52087) (28.0,1.52672) (29.0,1.532) };
    \addlegendentry{$\mathcal{S}$};
    \addplot coordinates { (9.0,1.00335) (10.0,1.0624) (11.0,1.11138) (12.0,1.15448) (13.0,1.18752) (14.0,1.21779) (15.0,1.24343) (16.0,1.26743) (17.0,1.2875) (18.0,1.30518) (19.0,1.32057) (20.0,1.33474) (21.0,1.34813) (22.0,1.35965) (23.0,1.37159) (24.0,1.38038) (25.0,1.39069) (26.0,1.39843) (27.0,1.40665) (28.0,1.41364) (29.0,1.41987) };
    \addlegendentry{$\mathcal{Y}'$};
    \addplot coordinates { (9.0,1.09666) (10.0,1.1492) (11.0,1.19097) (12.0,1.22702) (13.0,1.25576) (14.0,1.28303) (15.0,1.30494) (16.0,1.32615) (17.0,1.34338) (18.0,1.35905) (19.0,1.3717) (20.0,1.38479) (21.0,1.39692) (22.0,1.40663) (23.0,1.41684) (24.0,1.42478) (25.0,1.43189) (26.0,1.44075) (27.0,1.4474) (28.0,1.45325) (29.0,1.45947) };
    \addlegendentry{$\mathcal{C}$};
    \addplot coordinates { (9.0,1.15571) (10.0,1.20605) (11.0,1.24103) (12.0,1.27446) (13.0,1.29945) (14.0,1.32383) (15.0,1.3434) (16.0,1.36201) (17.0,1.37719) (18.0,1.39099) (19.0,1.40213) (20.0,1.4137) (21.0,1.42453) (22.0,1.43305) (23.0,1.44214) (24.0,1.44908) (25.0,1.45526) (26.0,1.46325) (27.0,1.46912) (28.0,1.47424) (29.0,1.47947) };
    \addlegendentry{$\mathcal{SP}$};
    \addplot coordinates { (9.0,1.20888) (10.0,1.26422) (11.0,1.30622) (12.0,1.3411) (13.0,1.36989) (14.0,1.4003) (15.0,1.4237) (16.0,1.44598) (17.0,1.4629) (18.0,1.48043) (19.0,1.49236) (20.0,1.50747) (21.0,1.52045) (22.0,1.52821) (23.0,1.54085) (24.0,1.54771) (25.0,1.55622) (26.0,1.563) (27.0,1.57188) (28.0,1.57542) (29.0,1.58435) };
    \addlegendentry{$\mathcal{L}$};
    \addplot coordinates { (9.0,0.994753) (10.0,1.04904) (11.0,1.09477) (12.0,1.13468) (13.0,1.16718) (14.0,1.19454) (15.0,1.21829) (16.0,1.23973) (17.0,1.25899) (18.0,1.27533) (19.0,1.29015) (20.0,1.30379) (21.0,1.31563) (22.0,1.32768) (23.0,1.33792) (24.0,1.34711) (25.0,1.3558) (26.0,1.36387) (27.0,1.37057) (28.0,1.3783) (29.0,1.38433) };
    \addlegendentry{$\mathcal{L}'$};
\end{axis}
\end{tikzpicture}
    \caption{Average comparison count (scaled by $n \ln n$) needed to sort a random
        input of up to $n = 2^{29}$ integers. We compare classical quicksort
        ($\mathcal{QS}$) with the Median-of-$3$ strategy, Yaroslavskiy's
        algorithm ($\mathcal{Y}$), the sampling algorithm ($\mathcal{SP}$), the
        counting algorithm
        ($\mathcal{C}$), the modified
        version of Sedgewick's algorithm from Section~\ref{sec:methods}
        ($\mathcal{S}$), and algorithm $\mathcal{L}$. Each of these dual-pivot
        algorithms uses the tertiles in a sample of size $5$ as the two pivots.
        Moreover, $\mathcal{L}'$ is an implementation of strategy $\mathcal{L}$
        which uses the third- and sixth-largest element from a sample of size
        $11$. $\mathcal{Y}'$ is Yaroslavskiy's algorithm choosing the tertiles
        of a sample of size $11$ as the two pivots. Each data point is the
        average over $400$ trials. }
    \label{fig:comp:sample}
\end{figure}
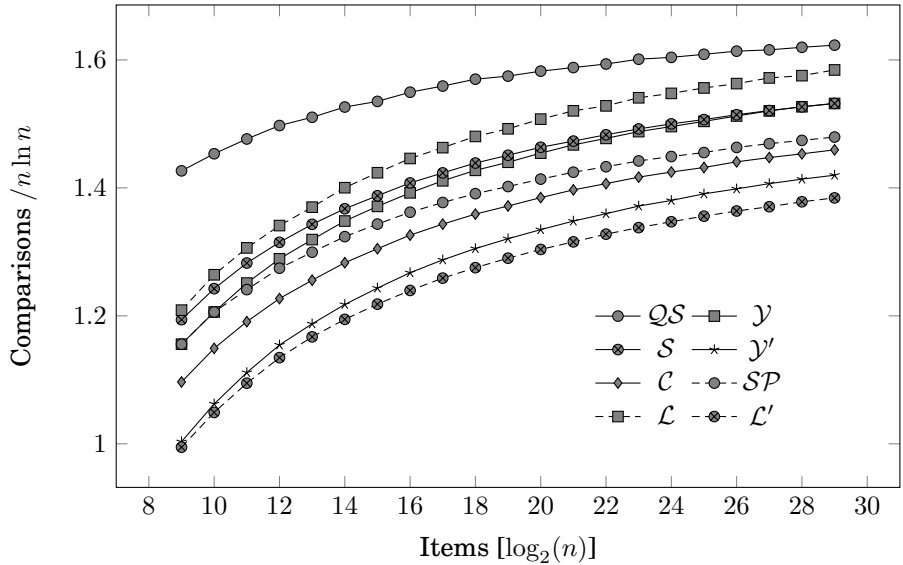

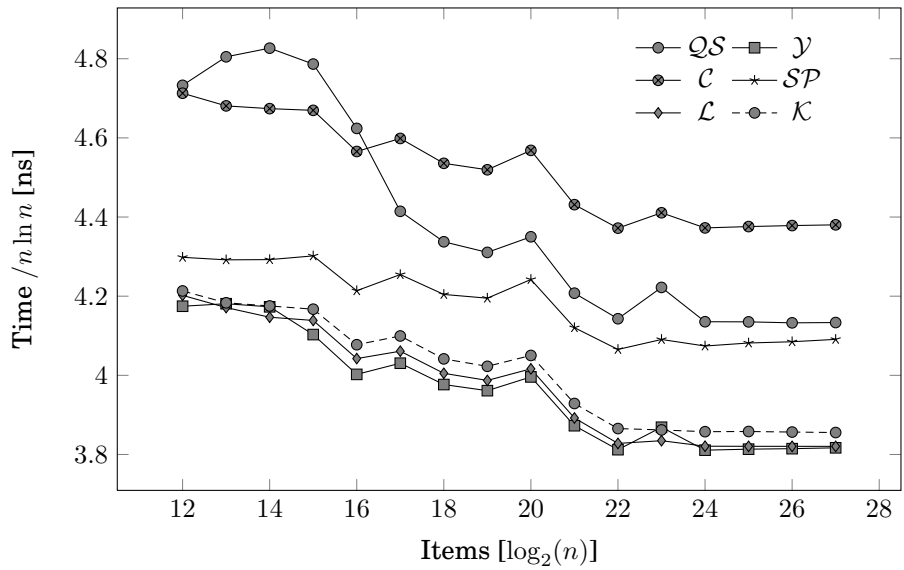
\begin{figure}
    \centering
\begin{tikzpicture}
  \begin{axis}[
    xlabel={Items [$\log_2(n)$]},
    ylabel={Time $/ n \ln n$ [ns]},
    height=8cm,
    width=12cm,
    legend style = { at = {(0.65,0.85)}, anchor=west, draw=none},
    cycle list name = black white,
    legend columns = 2
    ]
    % -IMPORT-DATA stats res/kt59-cpp-dual-optimized.txt
    % -MULTIPLOT(algo) SELECT LOG(2,items) AS x, AVG( time )/ (items * LN(items)) * 1e9 AS y, MULTIPLOT FROM stats GROUP BY MULTIPLOT,x ORDER BY MULTIPLOT,x
    \addplot coordinates {  (12.0,4.7328) (13.0,4.80484) (14.0,4.82673) (15.0,4.78657) (16.0,4.62393) (17.0,4.41442) (18.0,4.33746) (19.0,4.31062) (20.0,4.35008) (21.0,4.2076) (22.0,4.14301) (23.0,4.22214) (24.0,4.13544) (25.0,4.1351) (26.0,4.13266) (27.0,4.13339) };
    \addlegendentry{$\mathcal{QS}$};
    \addplot coordinates { (12.0,4.17461) (13.0,4.18093) (14.0,4.17313) (15.0,4.10297) (16.0,4.0024) (17.0,4.03064) (18.0,3.97671) (19.0,3.96124) (20.0,3.99598) (21.0,3.87268) (22.0,3.81213) (23.0,3.86882) (24.0,3.81088) (25.0,3.81342) (26.0,3.81468) (27.0,3.81686) };
    \addlegendentry{$\mathcal{Y}$};
    \addplot coordinates {  (12.0,4.71282) (13.0,4.68086) (14.0,4.67387) (15.0,4.66958) (16.0,4.5656) (17.0,4.59854) (18.0,4.53572) (19.0,4.51942) (20.0,4.56842) (21.0,4.43126) (22.0,4.37192) (23.0,4.41077) (24.0,4.37252) (25.0,4.37589) (26.0,4.37843) (27.0,4.38028) };
    \addlegendentry{$\mathcal{C}$};
    \addplot coordinates { (12.0,4.2982) (13.0,4.29183) (14.0,4.29243) (15.0,4.30182) (16.0,4.21382) (17.0,4.25464) (18.0,4.2044) (19.0,4.19495) (20.0,4.24222) (21.0,4.12046) (22.0,4.06537) (23.0,4.09061) (24.0,4.07407) (25.0,4.0817) (26.0,4.08485) (27.0,4.09089) };
    \addlegendentry{$\mathcal{SP}$};
    \addplot coordinates { (12.0,4.20207) (13.0,4.1712) (14.0,4.14683) (15.0,4.1387) (16.0,4.04234) (17.0,4.06099) (18.0,4.00508) (19.0,3.98714) (20.0,4.01641) (21.0,3.89226) (22.0,3.82797) (23.0,3.83479) (24.0,3.8209) (25.0,3.82063) (26.0,3.82055) (27.0,3.82052) };
    \addlegendentry{$\mathcal{L}$};
    \addplot coordinates {  (12.0,4.21324) (13.0,4.18325) (14.0,4.17531) (15.0,4.16714) (16.0,4.07734) (17.0,4.09946) (18.0,4.04156) (19.0,4.02275) (20.0,4.05023) (21.0,3.92884) (22.0,3.8655) (23.0,3.86197) (24.0,3.85758) (25.0,3.85806) (26.0,3.85668) (27.0,3.85547) };
    \addlegendentry{$\mathcal{K}$};

\end{axis}
\end{tikzpicture}
\caption{Running time experiments in {\Cpp}, setting 1 with compiler flags: \emph{-O2}. Each data point is
the average over 1000 trials. Times are scaled by $n \ln n$.}
\label{fig:running:time:setting:1}
\end{figure}

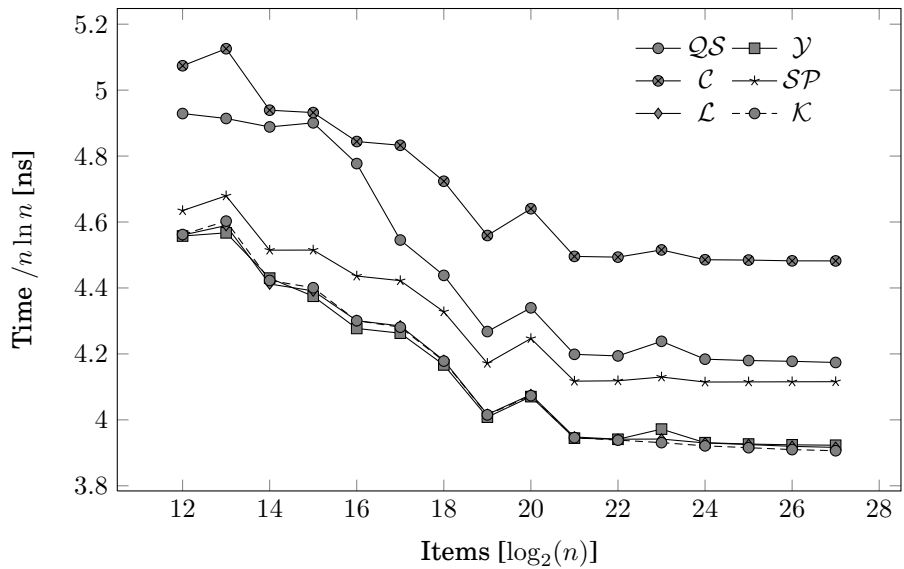
\begin{figure}
    \centering
\begin{tikzpicture}
  \begin{axis}[
    xlabel={Items [$\log_2(n)$]},
    ylabel={Time $/ n \ln n$ [ns]},
    height=8cm,
    width=12cm,
    legend style = { at = {(0.65,0.85)}, anchor=west, draw=none},
    cycle list name = black white,
    legend columns = 2
    ]
    % -IMPORT-DATA stats2 res/kt59-cpp-dual-optimized-unroll.txt
    % -MULTIPLOT(algo) SELECT LOG(2,items) AS x, AVG( time )/ (items * LN(items)) * 1e9 AS y, MULTIPLOT FROM stats2 GROUP BY MULTIPLOT,x ORDER BY MULTIPLOT,x
    \addplot coordinates {  (12.0,4.92893) (13.0,4.91397) (14.0,4.88841) (15.0,4.90075) (16.0,4.77715) (17.0,4.54535) (18.0,4.43814) (19.0,4.26778) (20.0,4.33975) (21.0,4.19873) (22.0,4.1939) (23.0,4.23816) (24.0,4.184) (25.0,4.17993) (26.0,4.17757) (27.0,4.17404) };
    \addlegendentry{$\mathcal{QS}$};
    \addplot coordinates {  (12.0,4.55754) (13.0,4.56718) (14.0,4.42987) (15.0,4.37486) (16.0,4.2772) (17.0,4.26282) (18.0,4.16655) (19.0,4.00821) (20.0,4.07047) (21.0,3.94463) (22.0,3.94113) (23.0,3.97194) (24.0,3.93071) (25.0,3.92669) (26.0,3.92428) (27.0,3.92304) };
    \addlegendentry{$\mathcal{Y}$};
    \addplot coordinates {  (12.0,5.07395) (13.0,5.12557) (14.0,4.93912) (15.0,4.93189) (16.0,4.84409) (17.0,4.83273) (18.0,4.72362) (19.0,4.55909) (20.0,4.64028) (21.0,4.49571) (22.0,4.49348) (23.0,4.51536) (24.0,4.4856) (25.0,4.48459) (26.0,4.48226) (27.0,4.48206) };
    \addlegendentry{$\mathcal{C}$};
    \addplot coordinates {  (12.0,4.63465) (13.0,4.67941) (14.0,4.51464) (15.0,4.51489) (16.0,4.43613) (17.0,4.42222) (18.0,4.32784) (19.0,4.17154) (20.0,4.24671) (21.0,4.11733) (22.0,4.11852) (23.0,4.1299) (24.0,4.11465) (25.0,4.11507) (26.0,4.11543) (27.0,4.11576) };
    \addlegendentry{$\mathcal{SP}$};
    \addplot coordinates {  (12.0,4.56033) (13.0,4.58935) (14.0,4.41201) (15.0,4.39107) (16.0,4.30099) (17.0,4.28473) (18.0,4.18026) (19.0,4.01642) (20.0,4.07635) (21.0,3.9481) (22.0,3.94131) (23.0,3.94176) (24.0,3.92996) (25.0,3.92507) (26.0,3.91964) (27.0,3.91647) };
    \addlegendentry{$\mathcal{L}$};
    \addplot coordinates {  (12.0,4.56197) (13.0,4.60267) (14.0,4.42271) (15.0,4.40071) (16.0,4.30075) (17.0,4.28105) (18.0,4.17829) (19.0,4.01525) (20.0,4.07339) (21.0,3.94594) (22.0,3.93846) (23.0,3.93093) (24.0,3.92121) (25.0,3.9156) (26.0,3.90977) (27.0,3.90636) };
    \addlegendentry{$\mathcal{K}$};

\end{axis}
\end{tikzpicture}
\caption{Running time experiments in \Cpp, setting $2$ with compiler flags: \emph{-O2 -funroll-loops}. Each data point is
the average over 1000 trials. Times are scaled by $n \ln n$.}
\label{fig:running:time:setting:2}
\end{figure}

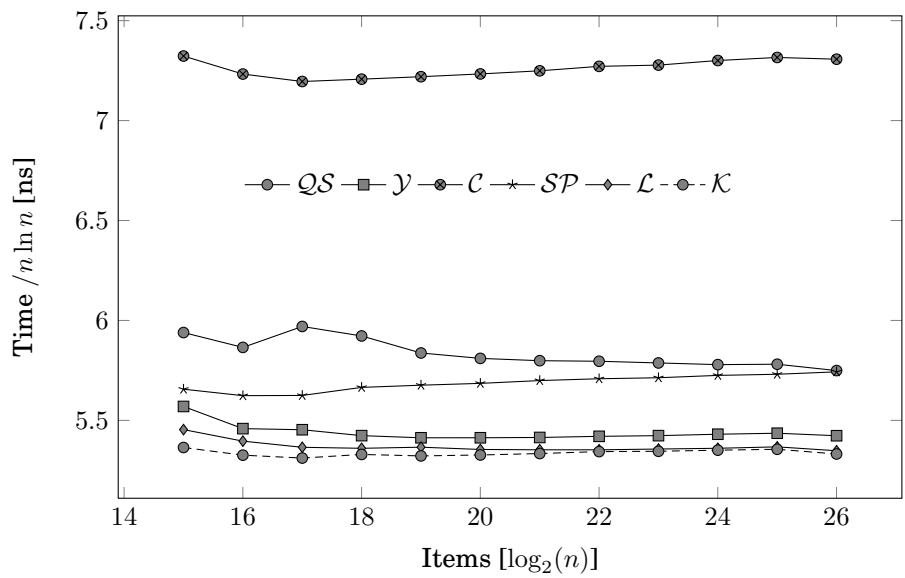
\begin{figure}
    \centering
\begin{tikzpicture}
  \begin{axis}[
    xlabel={Items [$\log_2(n)$]},
    ylabel={Time $/ n \ln n$ [ns]},
    height=8cm,
    width=12cm,
    legend style = { at = {(0.15,0.65)}, anchor=west, draw=none}, legend columns = 6,
    cycle list name = black white
    ]
    % -IMPORT-DATA stats5 res/kt59-java8-with-warmup-optimized.txt
    % -MULTIPLOT(algo) SELECT LOG(2,items) AS x, AVG( time )/ (items * LN(items)) * 1e9 AS y, MULTIPLOT FROM stats5 GROUP BY MULTIPLOT,x ORDER BY MULTIPLOT,x
    \addplot coordinates { (15.0008,5.93946) (16.0008,5.86508) (17.0008,5.9699) (18.0008,5.92234) (19.0008,5.8372) (20.0008,5.81015) (21.0008,5.79861) (22.0008,5.79579) (23.0008,5.78734) (24.0008,5.77896) (25.0008,5.78091) (26.0008,5.7492) };
    \addlegendentry{$\mathcal{QS}$};
    \addplot coordinates { (15.0008,5.56956) (16.0008,5.45866) (17.0008,5.45317) (18.0008,5.42361) (19.0008,5.41289) (20.0008,5.41302) (21.0008,5.41448) (22.0008,5.42014) (23.0008,5.4236) (24.0008,5.43069) (25.0008,5.43571) (26.0008,5.42317) };
    \addlegendentry{$\mathcal{Y}$};
    \addplot coordinates { (15.0008,7.32291) (16.0008,7.23281) (17.0008,7.19554) (18.0008,7.20725) (19.0008,7.21956) (20.0008,7.23337) (21.0008,7.24906) (22.0008,7.27134) (23.0008,7.27759) (24.0008,7.30056) (25.0008,7.31583) (26.0008,7.30714) };
    \addlegendentry{$\mathcal{C}$};
    \addplot coordinates { (15.0008,5.65602) (16.0008,5.62368) (17.0008,5.62476) (18.0008,5.66502) (19.0008,5.67603) (20.0008,5.68516) (21.0008,5.69917) (22.0008,5.70824) (23.0008,5.71327) (24.0008,5.72505) (25.0008,5.73076) (26.0008,5.74282) };
    \addlegendentry{$\mathcal{SP}$};
    \addplot coordinates { (15.0008,5.45384) (16.0008,5.39612) (17.0008,5.3653) (18.0008,5.36028) (19.0008,5.36584) (20.0008,5.35446) (21.0008,5.35281) (22.0008,5.35267) (23.0008,5.35661) (24.0008,5.36039) (25.0008,5.3675) (26.0008,5.34803) };
    \addlegendentry{$\mathcal{L}$};
    \addplot coordinates { (15.0008,5.36427) (16.0008,5.32594) (17.0008,5.31145) (18.0008,5.32954) (19.0008,5.3222) (20.0008,5.32684) (21.0008,5.33407) (22.0008,5.34376) (23.0008,5.34538) (24.0008,5.35092) (25.0008,5.35541) (26.0008,5.33182) };
    \addlegendentry{$\mathcal{K}$};

\end{axis}
\end{tikzpicture}
\caption{Running time experiments in Java 8. For warming up the JIT, we let each algorithm sort 10{\,}000
inputs consisting of $100{\,}000$ elements. Each
data point is the average over $500$ trials. Times are scaled by $n\ln n$.}
\label{fig:running:time:java:8}
\end{figure}

\end{document}